\documentclass[preprint,12pt]{elsarticle}

\pdfoutput=1
\usepackage{graphicx}
\usepackage{epsfig}
\usepackage{color}
\usepackage{ulem}
\usepackage{multirow}
\usepackage{amsmath, amsfonts}
\usepackage{amssymb}
\usepackage{latexsym}
\usepackage{hyperref}
\usepackage{ifthen}
\graphicspath{{figures/}}

\usepackage{amsthm}

\newtheorem{remark}{Remark}[section]
\newtheorem{theorem}{Theorem}[section]
\newtheorem{corollary}{Corollary}[section]

\journal{.}

\begin{document}

\begin{frontmatter}


 \author{Hwayeon Ryu\corref{cor1}}
 \ead{hryu@hartford}
 \cortext[cor1]{Corresponding author}
 \address{Department of Mathematics\\University of Hartford, West Hartford, CT 06117, USA}
\author{Sue Ann Campbell}
 \ead{sacampbell@uwaterloo.ca}
 \address{Department of Applied Mathematics and Centre for Theoretical Neuroscience\\University of Waterloo, Waterloo, Ontario, N2L 3G1, Canada}

\title{Geometric Analysis of Synchronization in Neuronal Networks with Global Inhibition and Coupling Delays}

\begin{abstract}
We study synaptically coupled neuronal networks to identify the role of coupling delays in network's synchronized behaviors. We consider a network of excitable, relaxation oscillator neurons where two distinct populations, one excitatory and one inhibitory, are coupled and interact with each other. The excitatory population is uncoupled, while the inhibitory population is tightly coupled. A geometric singular perturbation analysis yields existence and stability conditions for synchronization states under different firing patterns between the two populations, along with formulas for the periods of such synchronous solutions. Our results demonstrate that the presence of coupling delays in the network promotes synchronization. Numerical simulations are conducted to supplement and validate analytical results. We show the results carry over to a model for spindle sleep rhythms in thalamocortical networks, one of the biological systems which motivated our study. The analysis helps to explain how coupling delays in either excitatory or inhibitory synapses contribute to producing synchronized rhythms. 
\end{abstract}


\end{frontmatter}

\section{Introduction}
Oscillatory behavior in neuronal networks has been one of the main subjects to better understand the central nervous system~\cite{Linas88,Jacklet89,steriade90,traub91,buzsaki94}. Examples of dynamic behaviors include synchronization~\cite{Golomb94,Luz}, in which each cell in the network fires at the same time, and clustering~\cite{GR94,RT00}, in which the entire population of cells breaks up into subpopulations or clusters; cells within a single population fire at the same time but are desynchronized from ones in different subpopulations. Much more complicated network behaviors~\cite{traub91,TL97,KL94}, such as traveling waves~\cite{KimBal95,DBMS96,Golomb96,Rinzel98,TEY01}, are also possible.  

Neurons are connected mainly via chemical synapses, the junction of two nerve cells, through which information from one neuron transmits to another neurons, resulting in synaptic coupling. For this communication, the electrical signal must
travel along the axon of one neuron to the synapse, resulting in a {\it conduction delay}. The size of this delay depends on the diameter and length of the axon and whether or not it is myelinated \cite{tomasi2012}.  Further, 
once the electrical signal reaches the synapse, time is required for a neurotransmitter to be released and to travel through the synaptic cleft, a tiny gap between the nerve cells, and for the transmitter to cause an effect (through chemical reactions) on the postsynaptic cell. This time is called a {\it synaptic delay}. We call the combined effect of these two delays {\it coupling delay}. Synapses can be broadly classified in two types, excitatory and inhibitory each associated with particular neurons. Excitatory synapses tend to promote the transmission of electrical signals while inhibitory synapses tend to suppress the transmission. Although excitatory neurons are much more common in the brain \cite{douglas2007}, it has become increasingly apparent that inhibitory neurons play an important role in producing and regulating the behavior of brain networks \cite{roux2015}. Thus it is important to consider networks including both inhibitory and excitatory neurons.

The synaptic types, length of the delays, network connectivity and intrinsic properties of the neurons 
all interact to produce a variety of dynamic network behaviors, such as synchronization and clustering
\cite{GR94,BT,CampWang17,SSA11,MRTWBC,OroszSIADS14,choe2010,DLS12,Luz,KPR}.
Due to the richness of qualitatively different network behaviors caused by delays, the impacts of delays on such emergent network patterns are the key to understanding the information processing functions in the brain.
Many studies have been done on the effects of delays on networks where the synapses are exclusively excitatory or 
inhibitory \cite{CampWang98,FJWC01,Crook97,CampWang17,SSA11}. Here we address the role of delays in a network with both. There are many potential choices of network connectivity. We focus on a network with {\it global inhibition}, which consists of a uncoupled or sparsely coupled excitatory network reciprocally coupled to a highly connected inhibitory population. Networks with such structure are associated with rhythm generation in the CA1 region of hippocampus~\cite{bezaire2013} and the thalamus~\cite{DS97,DBMS96,CDSS97}, and with sensory processing~\cite{poo2009,doiron2003}. For the neural model, we focus on excitable, relaxation oscillators, the behavior of which is representative of many types of neurons. Our network may exhibit synchronous solutions and we prove sufficient conditions for the existence and stability of such solutions in terms of coupling delays. These results help to provide insight into how the intrinsic properties of individual cells interact with the synaptic properties, including coupling type and delays, to produce the emergent population rhythms. For example, we show that the presence of coupling delays may play a significant role in producing stable synchronous behaviors. 

We adapt geometric singular perturbation methods to analyze the mechanisms responsible for synchronization behaviors. The fundamental idea of this approach is to construct singular solutions by separating a system of differential equations into subsystems evolving on fast and slow time scales. Under some general hypotheses, actual solutions exist near these singular solutions. In the relaxation oscillator, the variables vary repeatedly between two distinct states corresponding to so-called {\it active} and {\it silent} phases. The amount of time spent in each phase substantially exceeds the time spent in the transitions between phases. When a relaxation oscillator is used to model a neuron, the rapid transition from the silent phase to the active phase corresponds {\it firing} of an action potential in the neuron.

Geometric singular perturbation approaches have been previously used to investigate the generation of pattern formation in neuronal networks~\cite{Somers93,skinner94,terman95,TL97,TK98,CampWang98,FJWC01,LoFaro99}. Despite the well-established results, most studies simplify their models to make mathematical analysis more tractable. The resulting simplified models lack key features: i) the direct interaction of coupling delays with intrinsic dynamics of neurons and ii) the underlying architecture of the network. For example, the effect of delay was considered in~\cite{CampWang98} but only for two neurons (i.e.,~not a network), and in~\cite{FJWC01} but for networks with a single type (excitatory) of neuron. Networks of relaxation oscillators involving both excitatory and inhibitory neurons are considered as in~\cite{terman95} but without coupling delays. More recently, Rubin and Terman considered the global inhibitory network, and analyzed the existence and stability of synchronous solutions~\cite{RUBIN02} and of clustered solutions~\cite{RT00}. However, their models have no conduction delay and the synaptic delay due to the chemical kinetics of the ion channel is implicitly included in the model for synaptic gating variable. To alleviate the model simplifications mentioned above while extending the previous studies~\cite{RT00,RT00b,RUBIN02}, we include an explicit representation of delays in model equations for the global inhibitory network, which will allow a systematic study for the delays in the context of network pattern formation.

Two important questions arise in the geometric analysis. The first is associated with the existence of a singular solution corresponding to synchronization. We assume that an individual cell, without synaptic input, is unable to oscillate. Thus, the existence of network synchronous behavior depends on whether the singular trajectory is able to ``escape" from the silent phase when they are coupled. The increased cellular or network complexity enhances each cell's opportunity to escape from the silent phase. The second question is concerned with the stability of the singular solution. To demonstrate the stability of a synchronous state, we need to show that the slightly perturbed trajectories of different cells are eventually brought closer together as they evolve in phase space. We show that this compression depends on the underlying network architecture as well as nontrivial interactions between the intrinsic and synaptic properties of the cells~\cite{TK98}. Our analysis shows, for example, how delays promote stable synchronized behaviors due to their interaction with intrinsic properties of neurons.

The reminder of the paper is organized as follows. In Section~\ref{model}, we present the models for individual relaxation oscillators and for the dynamic coupling between oscillators which will be used in our study. Also we describe the architecture of the global inhibitory network consisting of two distinct populations of oscillators; one population inhibits the other, which in turn excites the first population. Section~\ref{model} also introduces basic terminology needed for singular perturbation analysis, including the notion of a singular solution. In Section~\ref{analysis}, we present the statement and proof of existence and stability results under different conditions on the relative duration of the active phase between two populations. Section~\ref{numerical} follows to supplement our analytical results in Section~\ref{analysis} by illustrating the synchronous solutions obtained by numerical simulations. Also this section includes the numerical results for thalamic models motivated by thalamocortical networks~\cite{DMS93,SMS93,Golomb94,DS97,TBK96}. Finally, we conclude with a discussion in Section~\ref{discussion}.
\section{The Models}\label{model}
We describe the model equations corresponding to individual, uncoupled cells. There are two types: one for inhibitory cells and one for excitatory cells. Then, we introduce the synaptic coupling between the cells, delays, and network architecture to be considered. Finally, based on the model equations corresponding to the network, we consider fast and slow subsystems, which will be used for singular geometric analysis in subsequent sections.

\subsection{Single cells}
We model an individual cell of the networks as a relaxation oscillator, whose equations are given by
\begin{align} 
\label{singx}
\dot x&=f(x,y),\\
\label{singy}
\dot y&=\epsilon g(x, y),
\end{align}
where $^.={d\over dt}$, $x \in \mathbb{R}$, and $y\in \mathbb{R}^n$. For simplicity, we consider $n=1$ in our analysis (see \cite{RT00b} for an example with $n>1$). Here we assume $0<\epsilon \ll 1$ for singular geometric analysis so that $x$ is a fast variable and $y$ is a slow variable. Also, we assume that the $x$-nullcline, $f(x,y)=0$, is a cubic function, with left, middle, right branches, and $f>0$ ($f<0$) above (below) the $x$-nullcline curve. In addition, the $y$-nullcline is assumed to be a monotone decreasing function that intersects $f=0$ at a unique fixed point, and $g>0$ ($g<0$) below (above) the $y$-nullcline curve. See Figure~\ref{single_orbit}.

\begin{figure} [htb!]
\centering
  \includegraphics[height=60mm, width=90mm]{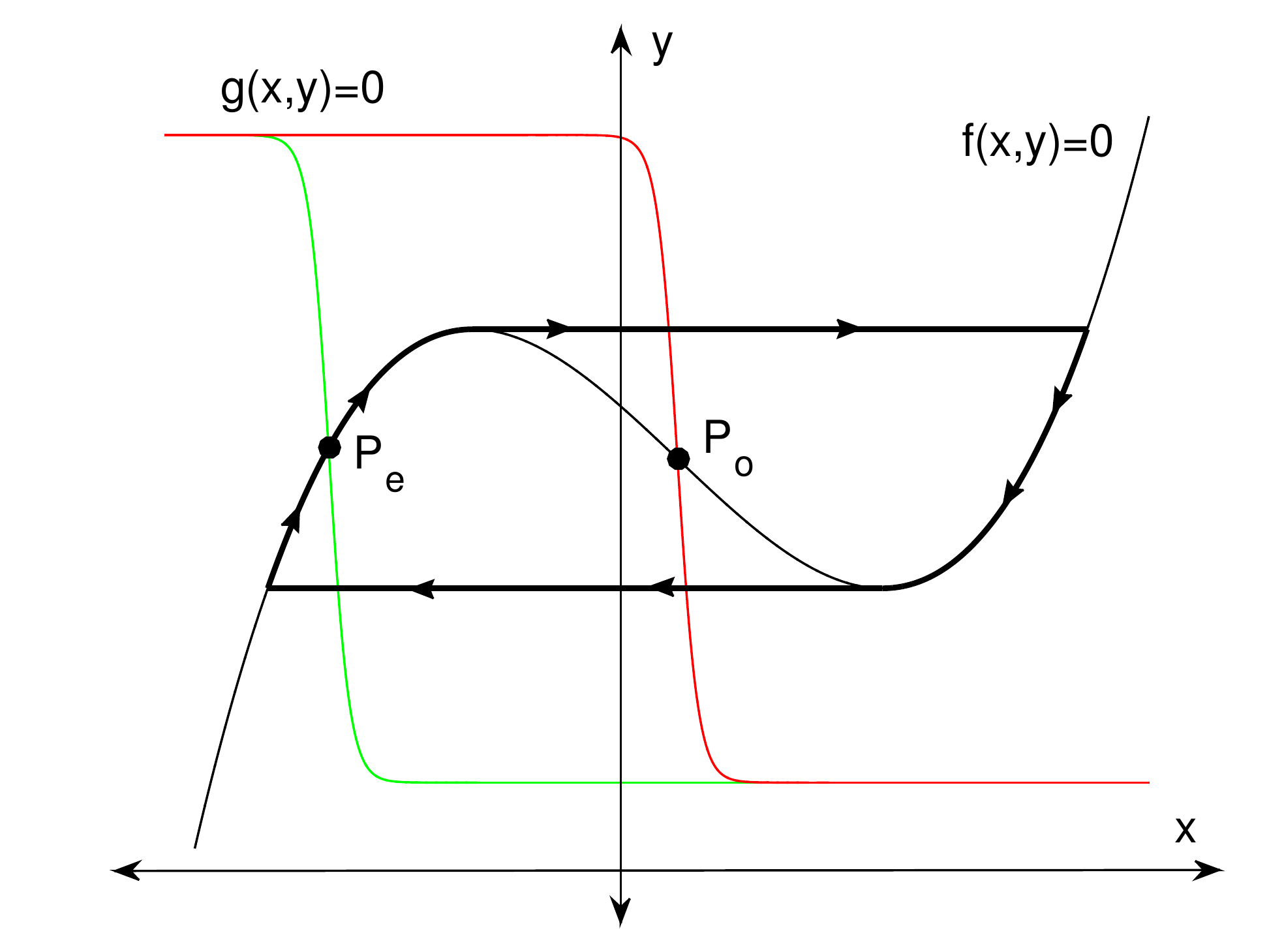}
  \caption{Nullclines for Eqs.~(\ref{singx})--(\ref{singy}) in both excitable (green line) and oscillatory (red line) cases. $P_e$ and $P_o$ correspond to the unique fixed points for excitable and oscillatory systems, respectively. The solid line shows a singular periodic solution for the oscillatory system.}
  \label{single_orbit}
\end{figure}

Depending on the location of the fixed point along the $x$-nullcline, we have two different situations: (i) the system is {\it excitable} if the fixed point lies on the left branch of $f=0$, as labeled $P_e$ in Fig.~\ref{single_orbit}; (ii) the system is {\it oscillatory} if the fixed point lies on the middle branch of $f=0$, labeled $P_o$. For the excitable system, $P_e$ is a stable fixed point, and no periodic solutions arise for all small $\epsilon$. However, if a sufficient amount of input is applied to the excitable system, the solution can jump to the right branch of $f=0$ and remain there for some time before
returning to the fixed point $P_e$, in this case we say the neuron {\it fires} or generates an action potential. 
On the other hand, in the oscillatory system, Eqs.~(\ref{singx})--(\ref{singy}) yield a periodic solution for all sufficiently small $\epsilon$, as shown in Fig.~\ref{single_orbit}. Since the thalamic cells we model in this study are known to be excitable during the sleep state \cite{steriade90,SMS93,DS97}, we will focus on the excitable system in subsequent sections.

\subsection{Synaptic coupling and network architecture}
We consider networks with the architecture as shown in Fig.~\ref{network}, which are motivated by models for the thalamic sleep rhythms \cite{Destexhe98,Golomb94,wang92}. In this architecture, called a {\it globally inhibitory network}, two distinct populations of cells interact with each other. Specifically, $J$-cells inhibit $E$-cells, which, in turn, excite the $J$ population. However, there is no communication among $E$-cells. In the spindle rhythms, the cells within the $J$ population are completely synchronized, thus we can view the entire $J$ population as a single cell, sending inhibition to the $E$ population globally. We assume that all $E$-cells are identical, but differ from the $J$-cell. To simplify the analysis, we shall assume that there are only two cells in $E$ population but it can be easily generalized to the case of an arbitrary number of $E$-cells.

\begin{figure} [htb!]
\centering
  \includegraphics[height=30mm, width=70mm]{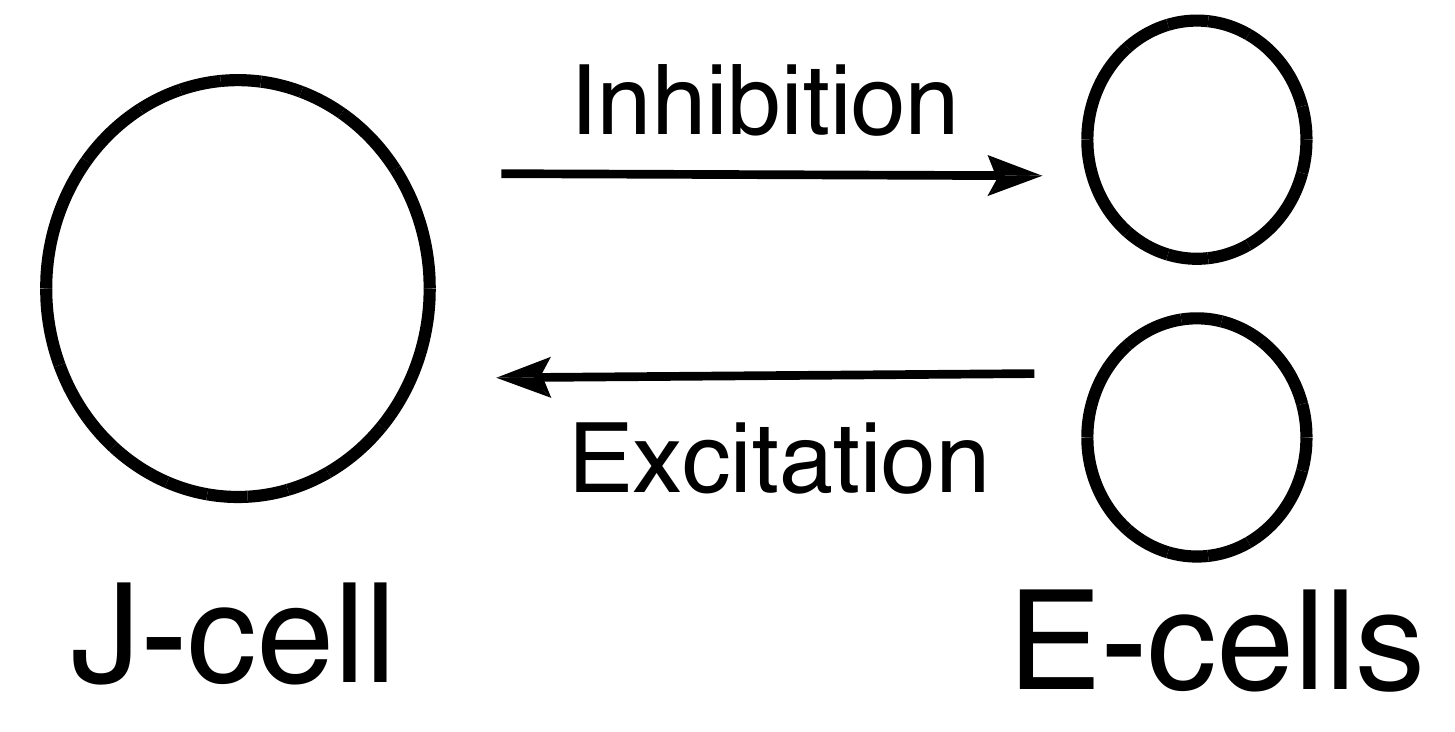}
  \caption{Schematic diagram of globally inhibitory network. The $J$-cell inhibits the $E$-cells, which, in turn, excite the $J$-cell.}
  \label{network}
\end{figure}

The equations corresponding to each $E_i$ for $i=1, 2$ in the network are
\begin{align} 
\label{excx}
\dot x_i&=f(x_i,y_i)-g_{inh}s_J(x_J(t-\tau_J))(x_i-x_{inh}),\\
\label{excy}
\dot y_i&=\epsilon g(x_i, y_i),
\end{align}
where $f$ and $g$ are defined as in Eqs.~(\ref{singx})--(\ref{singy}), and $g_{inh}>0$ represents the maximal conductance of the synapse, which can be viewed as coupling strength from the $J$-cell to each $E$-cell. The function $s_J$ determines the {\it inhibitory} synaptic coupling from $J$ to $E$. It is a sigmoidal
function which takes values in $[0,1]$.
Since the $J$-cell sends inhibition to the $E$-cells, $x_{inh}$, the reversal potential for the synaptic connection, is set so that $x_i-x_{inh}>0$. Finally, $\tau_J$ denotes the delay in the inhibitory synapse.  

The model equations for $J$ are similarly given by
\begin{align} 
\label{inhx}
\dot x_J&=f_J(x_J,y_J)-g_{exc}\left({1\over N}\sum_{i} s_i(x_i(t-\tau_E))\right)(x_J-x_{exc}),\\
\label{inhy}
\dot y_J&=\epsilon g_J(x_J, y_J),
\end{align}
where $g_{exc}$ denotes the maximal conductance of the {\it excitatory} synapse from $E$ to $J$. As in the model for the $E$ cell, the $s_i$ are sigmoidal 
functions with values in $[0,1]$.
The reversal potential for the {\it excitatory} synapse, denoted by $x_{exc}$, is chosen so that $x_J-x_{exc}<0$. The delay in the excitatory synapse, $\tau_E$, is assumed to be same for all the $E$-cells. For the case of two $E$-cells in the network, let us define $s_{tot}\equiv {1\over 2}(s_1+s_2)$.
Note that we do not incorporate chemical kinetics for synapses into our model. 
However, $\tau_E$ and $\tau_J$ include the effect of delays due to the chemical kinetics, as well as other factors.

Equations~\eqref{excx}--\eqref{inhy} form a four dimensional system of delay differential
equations. The appropriate initial data for such a system specifies functions 
for the variables on the interval $\tau\le t\le 0$, where 
$\tau=\max(\tau_E,\tau_J)$, yielding an infinite dimensional phase space.
In our analysis, however, we will assume that the synaptic functions $s_i$ and $s_J$ 
are Heaviside step functions, thus the values switch between 0 and 1 at the 
threshold $x$-value. The system~\eqref{excx}--\eqref{inhy} then becomes a discontinuous 
or switched system of ordinary differential equations, with a delayed 
switching manifold. That is, at any time the system evolves according to the
ODEs given by Eqs.~\eqref{excx}--\eqref{inhy} with the each of the $s_i$ and $\ s_J$ either 
$0$ or $1$, but the condition that determines which system of ODEs is followed
depends on the delayed values of $x_i$ and $x_J$. While there is a fairly large
literature on the stability of such systems (see e.g.,~\cite{Fridman2002,Sun2006}), the bifurcation theory of such systems is still being developed, with many results to date based on direct analysis of specific systems 
\cite{Sieber2006,Barton2006,Sieber2010}, such as what we will carry out. 
In our numerical simulations we will take the synaptic functions to be smooth,
approximations of Heaviside step functions.

\begin{remark}
An excitable cell stays at its stable fixed point unless it receives some synaptic input. The effect of this input depends on the type of coupling. For example, since $x_i-x_{inh}>0$, inhibitory coupling decreases $\dot x_i$, making it harder for the $E$-cells to fire. On the other hand, since $x_i-x_{exc}<0$ excitatory coupling increases $\dot x_J$, making it easier for the $J$-cell to fire. 
\end{remark}

The present model is similar to the model developed by Rubin and Terman~\cite{RUBIN02} in that both describe the dynamics of synaptic connection between two distinct populations in a globally inhibitory network. However, in their model, there are additional differential equations for the synaptic gating variables, $s_i$ and $s_J$. In these equations other slow variables are introduced which ensure the existence of synchronous solution. Our model, on the other hand, has no differential equations for the synaptic variables, and the synaptic coupling is a direct function of the appropriate $x$ variable. However, we include time delays in the connections, as in \cite{CampWang98,FJWC01}. Our model is different from that of \cite{CampWang98,FJWC01} as in their models the uncoupled neurons are oscillatory, instead of excitable. 

To conduct singular perturbation analysis, we identify the {\it fast} and {\it slow subsystems} for each population's evolution by dissecting the full system of equations given in Eqs.~(\ref{excx})--(\ref{inhy}).  The {\it fast subsystem} of Eqs.~(\ref{excx})--(\ref{inhy}) is obtained by simply setting $\epsilon=0$, which results in
\begin{align} 
\dot x_i&=f(x_i,y_i)-g_{inh}s_J(x_J(t-\tau_J))(x_i-x_{inh}),\\
\dot y_i&=0,\\
\dot x_J&=f_J(x_J,y_J)-g_{exc}\left({1\over N}\sum_{i} s_i(x_i(t-\tau_E))\right)(x_J-x_{exc}),\\
\dot y_J&=0,
\end{align}
where $^.={d\over dt}$. Note that the coupling between the fast systems of the $E$ and $J$ cells is only through delayed values of the $x$ variables.

The {\it slow subsystem} is derived by first introducing a slow time scale $\tilde t=\epsilon t$ and $\tilde \tau=\epsilon \tau$, and then setting $\epsilon=0$. This leads to a reduced system of equations for the slow variables only, after solving for each fast variable in terms of the slow ones. Let $x=\Phi_L(y,s)$ denote the left branch of the cubic $f(x,y)-g_{inh} s(x-x_{inh})=0$, and $G_L(y,s)\equiv g(\Phi_L(y,s),s)$. After dropping the tildes, we have the following 
equations
\begin{align} 
\label{slow_left}
x_i&=\Phi_L(y_i,s_J),\\
\label{slow_left1}
y_i'&=G_L(y_i,s_J),\\
\label{slow_synaptic}
s_J&=s_J(x_J(t-\tau_J)),
\end{align} 
where $'={d\over d\tilde t}$. The system in Eqs.~(\ref{slow_left})--(\ref{slow_synaptic}) determines the slow evolution of $E$-cell on the left branch. The slow subsystems of $E$-cell on the right branch and of $J$-cell on either branch can be similarly derived.

The slow subsystems determine the evolution of the $y$-variables in either the left branch (the silent phase) or the
right branch (the active phase). During this evolution, each cell travels along the left or right branch of some ``cubic" nullcline, which is determined by the total amount of synaptic input that the cell receives. A fast jump occurs when one of the cells reaches the left or right ``knee" of its corresponding cubic. Once reaching the knee, the cell may either jump up from the silent to the active phase or vice versa, depending on where the cell originally travels before jumping. For neuronal models, a fast jump from the low-$x$ branch (the silent phase) to the high-$x$ branch (the active phase) represents the generation of an action potential by a neuron. Thus, in the $\epsilon=0$ limit, we can construct a {\it singular solution} by connecting the solution to the slow subsystem with jumps between branches given by solutions to the fast subsystem. The analysis we provide in this study focuses on such singular solutions. For the extensions to small positive $\epsilon$, refer the work in \cite{Mish80,MPN86,CLMP89}. 

\begin{remark}
We analyze the dynamics of the network by constructing singular solutions. If $g_{inh}$ is not too large, then $f(x,y)-g_{inh} s_J(x-x_{inh})=0$ represents a cubic-shaped curve for each $s_J \in [0,1]$. Let us denote this curve by $C_{s_J}$; curves $C_0$ and $C_1$ are shown in Figure~\ref{nullclines_nodelay2}A. The trajectory for $E_i$ lies on the left/right branches of one of these curves during the silent/active phase, respectively. Fast jumps between different phases occur when an $E_i$ reaches a left or right knee of its respective cubic. Similarly, $J$ lies on the cubic curve determined by its total synaptic input $s_{tot}$, as shown in Fig.~\ref{nullclines_nodelay2}B. Note in Fig.~\ref{nullclines_nodelay2}A that the $s_J=1$ nullcline ($C_1$) lies above the $s_J=0$ nullcline ($C_0$), while in Fig.~\ref{nullclines_nodelay2}B, the $s_{tot}=1$ nullcline lies below the $s_{tot}=0$ nullcline. These relations result from the fact that the $E_i$ receives inhibition from $J$ while $J$ receives excitation from the $E_i$. 
\end{remark}

\begin{remark}
As mentioned earlier, one motivation for the global inhibitory model we consider
is the structure of thalamocortical networks~\cite{Destexhe98,Golomb94,wang92}. In these networks, there are two distinct but coupled populations of cells, thalamocortical relay cells corresponding to $E$-cells and thalamic reticular cells corresponding to $J$-cells in our model. 
\end{remark}
\section{Model Analysis}\label{analysis}
In this section, we first consider the case of no coupling delay to prove that the synchronous solution among $E$-cells does not exist provided in Section~\ref{no delay_longJ}. We then give sufficient conditions for the existence of a singular synchronous periodic solution, under different conditions on the duration of active phases for both populations. Depending on the relative duration of $J$-cell active phase to that of $E$-cells, we consider two cases: (i) the active phase of $J$-cell is longer than $E$-cells, given in Section~\ref{subsec_longerJ}; (ii) the active phase of $E$-cells is longer than $J$-cell, given in Section~\ref{subsec_longerE}. Also, we provide a brief stability analysis at the end of each case. In the following analysis, we denote the fixed point on the left branch of $C_{s_J}$ in Eqs.~(\ref{excx})--(\ref{excy}) by ($x_F(s_J), y_F(s_J))$, the left knee by ($x_L(s_J), y_L(s_J)$), and the right knee by ($x_R(s_J), y_R(s_J)$).

\subsection{Dynamics with no delay}\label{no delay_longJ}
We first show that an oscillatory synchronous solution among $E$-cells does not exist if there is no time delay in the interaction between the two populations, regardless of their starting points. 

\begin{theorem}
\label{theorem_nodelay_longJ}
Suppose that there is no delay in the synapses for the globally inhibitory network, i.e.,~$\tau_J=\tau_E=0$, in Eqs.~(\ref{excx})--(\ref{inhy}), and that $y_F(1)>y_L(0)$ is satisfied. Then, regardless of the the starting positions for both populations, a singular synchronous periodic solution does not exist.
\end{theorem}

\begin{proof}
We assume that the positions of $E$-cells are identical and show that no periodic solution exists as both the $E$-cells and $J$-cell always converge to their respective equilibrium points. We divide the proof into four cases, depending on the initial conditions, i.e.,~the starting points of the two cell types.

{\bf Case 1:} Suppose that both the $E$-cells and $J$-cell start in the
silent phase, i.e., the $E$ cells lie on the left branch of the $s_J=0$
nullcline and the $J$ cell lies on the left branch of the $s_{tot}=0$ nullcline.
Then both populations evolve according to their intrinsic dynamics (i.e.,~no coupling) and they will stay on their respective left branches and evolve towards their respective equilibrium points.

{\bf Case 2:} Suppose that the $E$-cells start in the active phase and the $J$
cell starts in the silent phase, i.e., the $E$ cells lie on the right branch
of $s_J=0$ nullcline and the $J$-cell on the left branch of the $s_{tot}=1$ nullcline (points $P_0$ and $Q_0$ in Figure~\ref{nullclines_nodelay2}). Then the $E$-cells will follow the right branch until they reach
the right knee at $P_1$ as shown in Fig.~\ref{nullclines_nodelay2}A, while the $J$ cells follow the left branch to the point $Q_1$. At this point the $E$-cells will jump down to the silent phase.  As the $J$-cell is already in the silent 
phase, the $E$-cells will jump down to the left branch of the $s_J=0$ nullcline, 
point $P_2$.  As soon as the $E$-cells cross the threshold $x$-value, excitation to the $J$-cell will turn off and the $J$-cell will jump to the left branch of the $s_{tot}=0$ nullcline, point $Q_2$.  We are now in Case 1.

\begin{figure} [htb!]
\centering
\includegraphics[height=52mm, width=58mm]{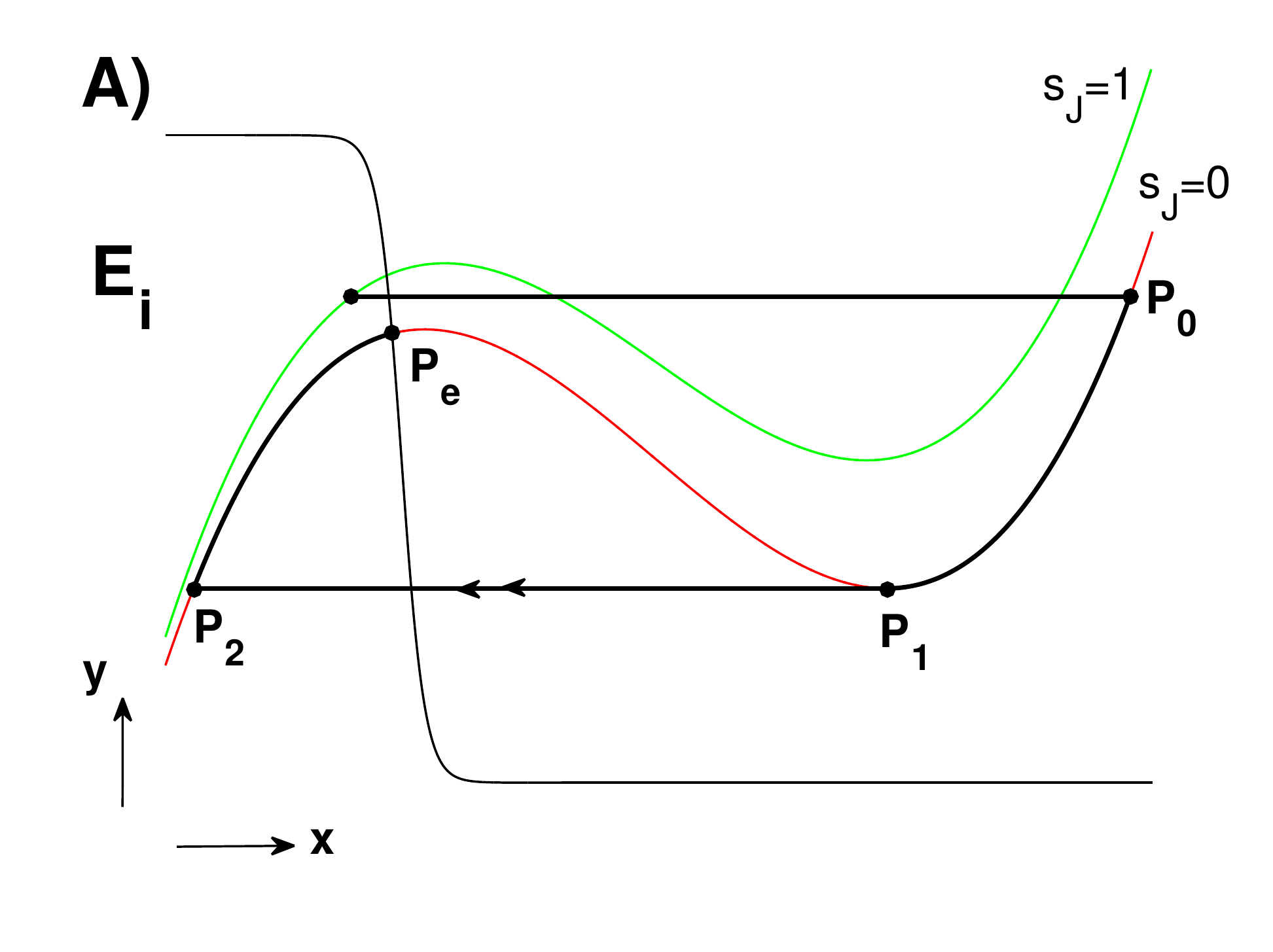}
  \includegraphics[height=52mm, width=58mm]{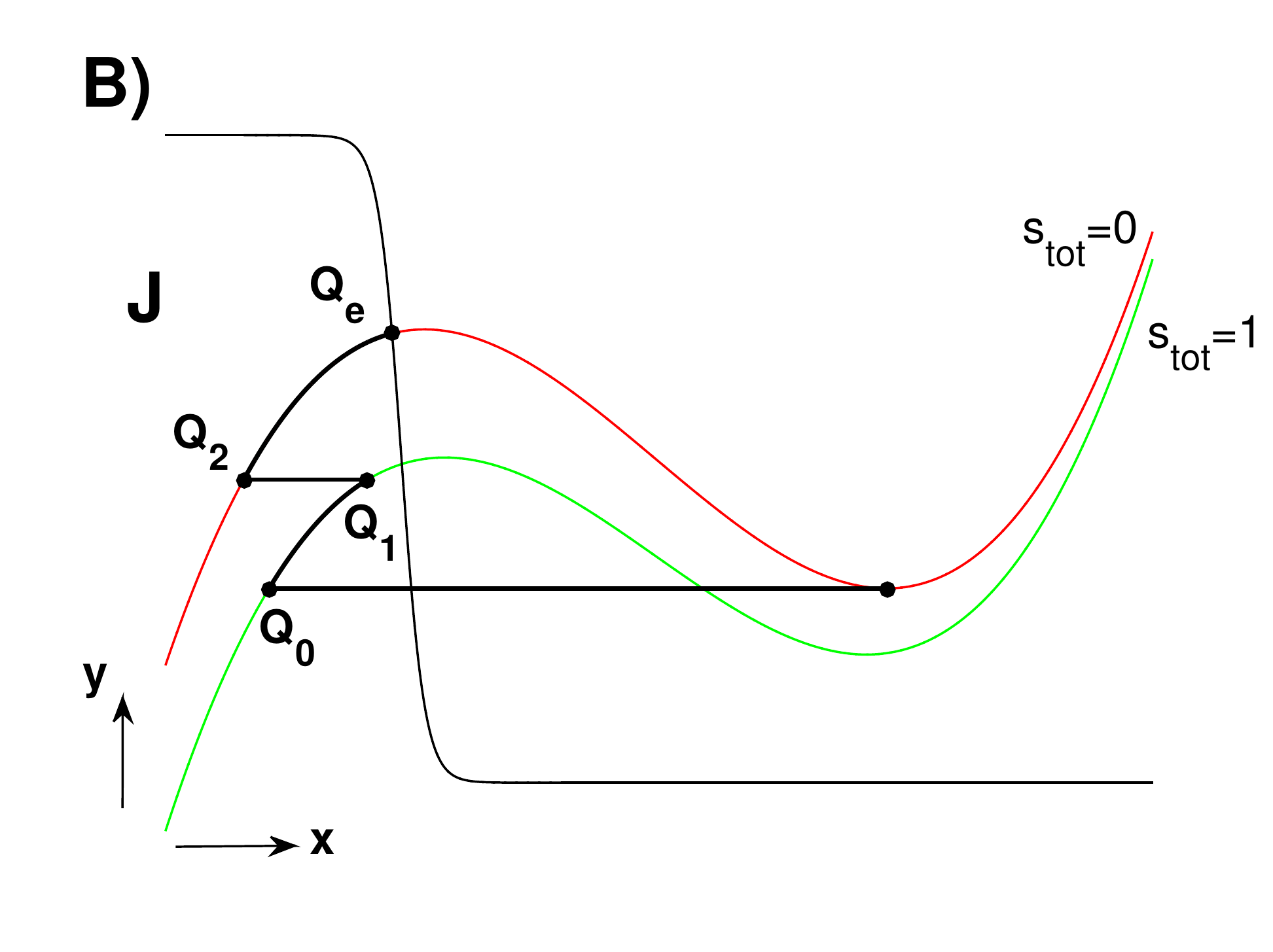}
  \caption{Nullclines for A) $E$-cells and B) $J$-cell in a globally inhibitory network. The solid lines, and points $P_i$ and $Q_i$ correspond to the singular synchronous solution constructed in the text. The double arrows on the solid lines indicate the fast jumps between the silent and active phases. 
The trajectories shown are for the case when the $E$-cell population starts in the active
phase, the $J$-cell population starts in the silent phase and there is no delay in the synapses.} 
  \label{nullclines_nodelay2}
\end{figure}

{\bf Case 3:} Suppose that both the $E$-cells and the $J$-cells start in the
active phase, i.e., the $E$-cells lie on the right branch of the  $s_J=1$
nullcline and the $J$-cell on the right branch of the $s_{tot}=1$ nullcline.
There are a number of possible solution trajectories for the cells, depending
on (i) which cell type reaches the right knee of its respective nullcline
first and (ii) the position the cell is in when the other cell type jumps down to the silent phase.  

One possible solution trajectory set is illustrated in Figure~\ref{nullclines_nodelay}, where $P_0$ and $Q_0$ correspond to their respective starting points. This figure corresponds to the case
where the $E$-cells jump down first. In the figure, the cells evolve
to points $P_1$ and $Q_1$, respectively, then $E$-cells jump down to $P_2$
on the left branch of the $s_J=1$ nullcline. As this occurs, the $J$-cell 
jumps to the left branch of the $s_{tot}=0$ curve, point $Q_2$. The $J$-cell follows the nullcline 
to the right knee, $Q_3$, and then jumps down to its silent phase. In the case shown in
Figure~\ref{nullclines_nodelay}A, when this occurs the $E$-cell lies at
point $P_3$ {\it below} the left knee of the $s_J=0$ nullcline. Thus
when the $J$-cell jumps down to $Q_4$ on the $s_{tot}=0$ nullcline, the $E$-cell
jumps to $P_4$ on the $s_J=0$ nullcline and we are in Case 1.

\begin{figure} [htb!]
\centering
\includegraphics[height=52mm, width=58mm]{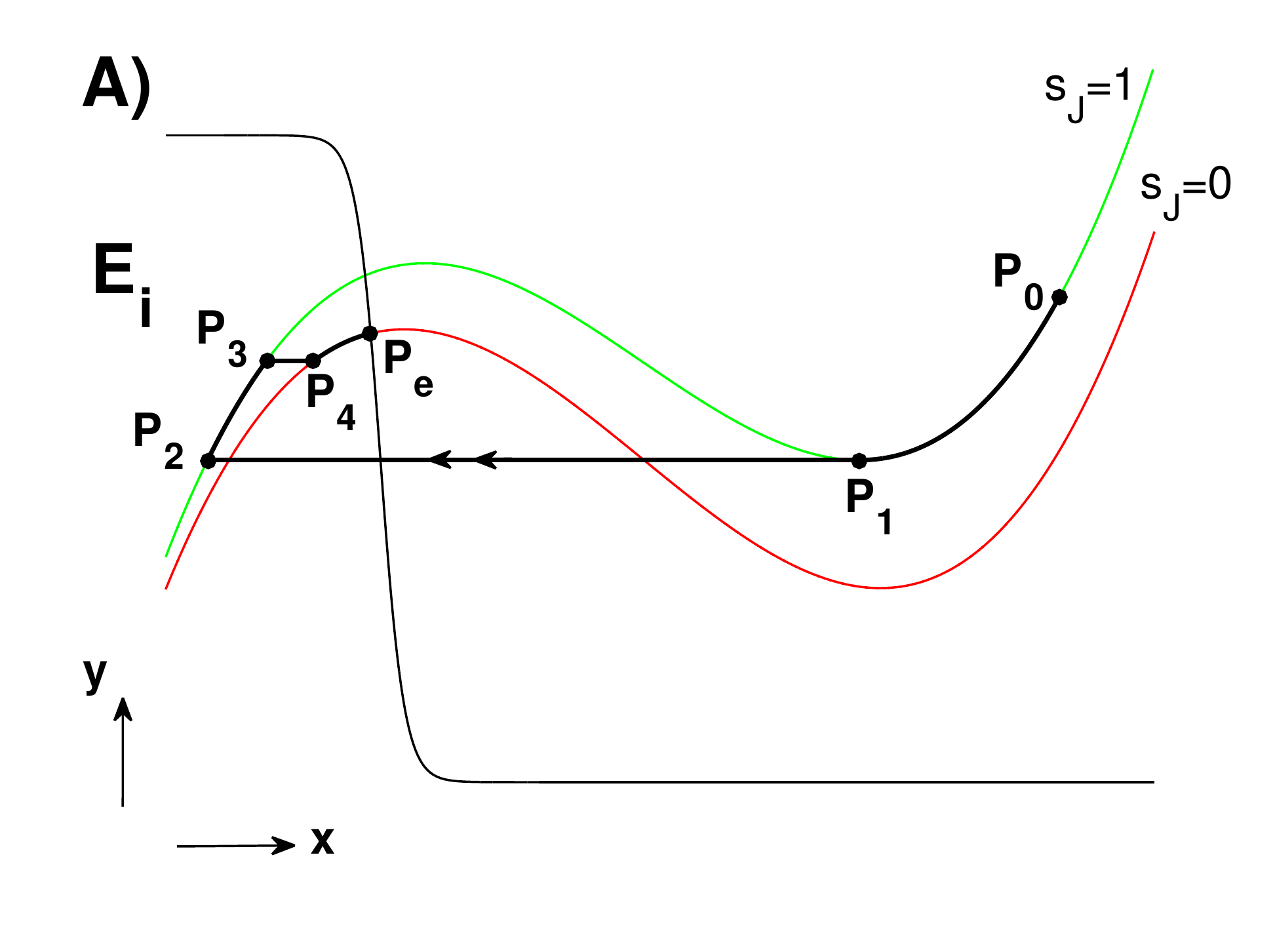}
  \includegraphics[height=52mm, width=58mm]{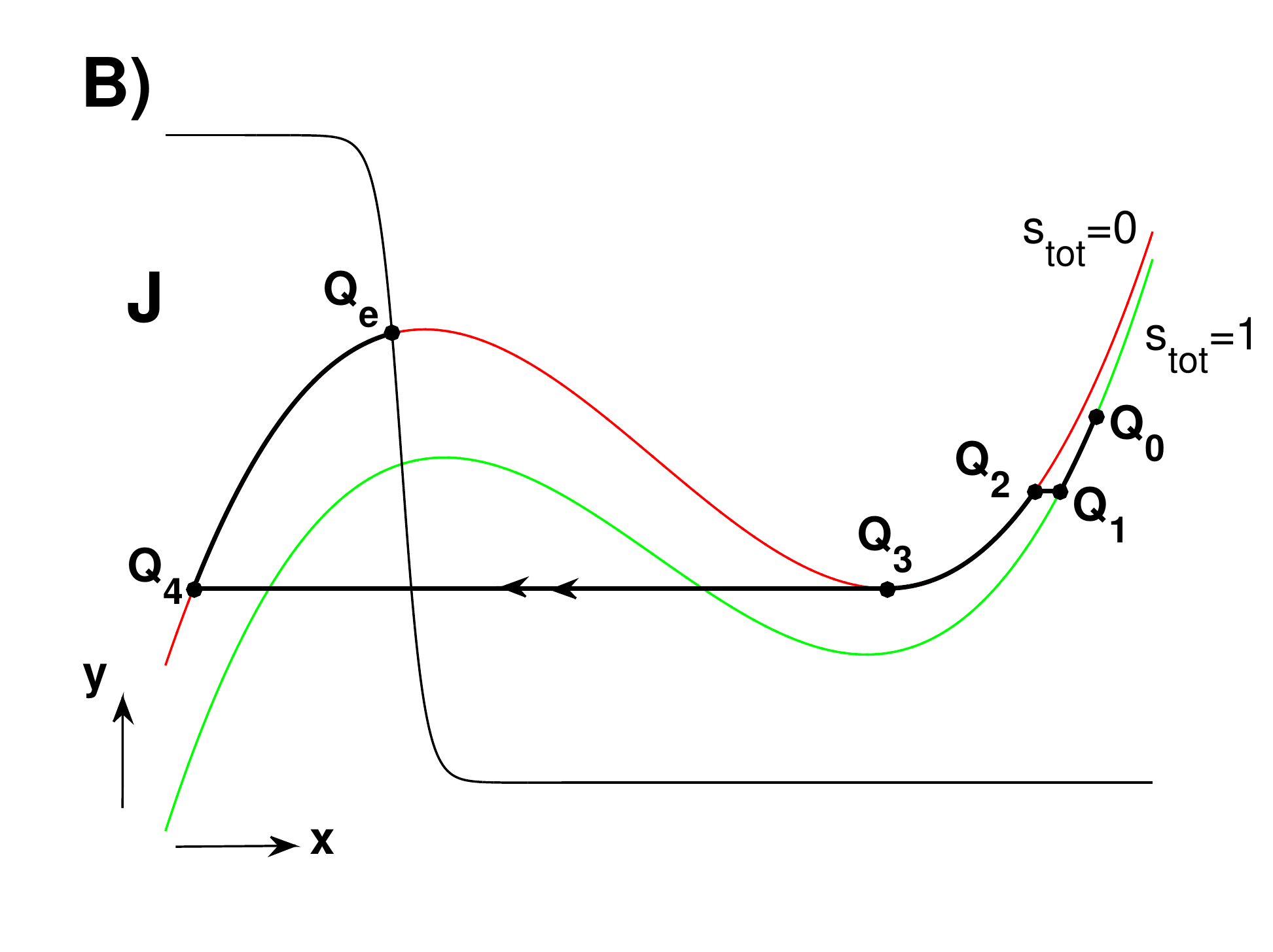}
  \caption{Plots of solution trajectories for A) $E$-cells and B) $J$-cell in black solid lines, approaching to their respective equilibrium points, if two populations start in active phase and there is no delay in the synapses. These plots are for the particular case where the $E$-cells lie below the left knee of the $s_J=0$ cubic in red curve of A) when the $J$-cell jumps down.}
  \label{nullclines_nodelay}
\end{figure}

A second solution trajectory set is shown in Figure~\ref{nullclines_nodelay1}. The only 
difference in this situation is that the point $P_3$ lies {\it above} the left knee of the 
$s_J=0$ nullcline. Thus when the $J$-cell jumps down from the active to the silent phase,
$Q_3$ to $Q_4$ on Fig.~\ref{nullclines_nodelay1}B, the $E$-cells jump up from the silent to the active phase as they are released from inhibition. This means that $Q_4$ is on the left branch of the $s_{tot}=1$ nullcline and $P_4$ is on the right branch of the $s_J=0$ nullcline. The two cells move along their respective nullclines until the $E$-cell reaches the right knee, $P_5$ on the Fig.~\ref{nullclines_nodelay1}A. The $E$-cells then jump down to $P_6$ on the left branch of the $s_J=0$ nullcline. This causes the $J$-cell to jump to $Q_6$ on the left branch of the $s_{tot}=0$ nullcline and we are in Case 1. 

\begin{figure} [htb!]
\centering
\includegraphics[height=52mm, width=58mm]{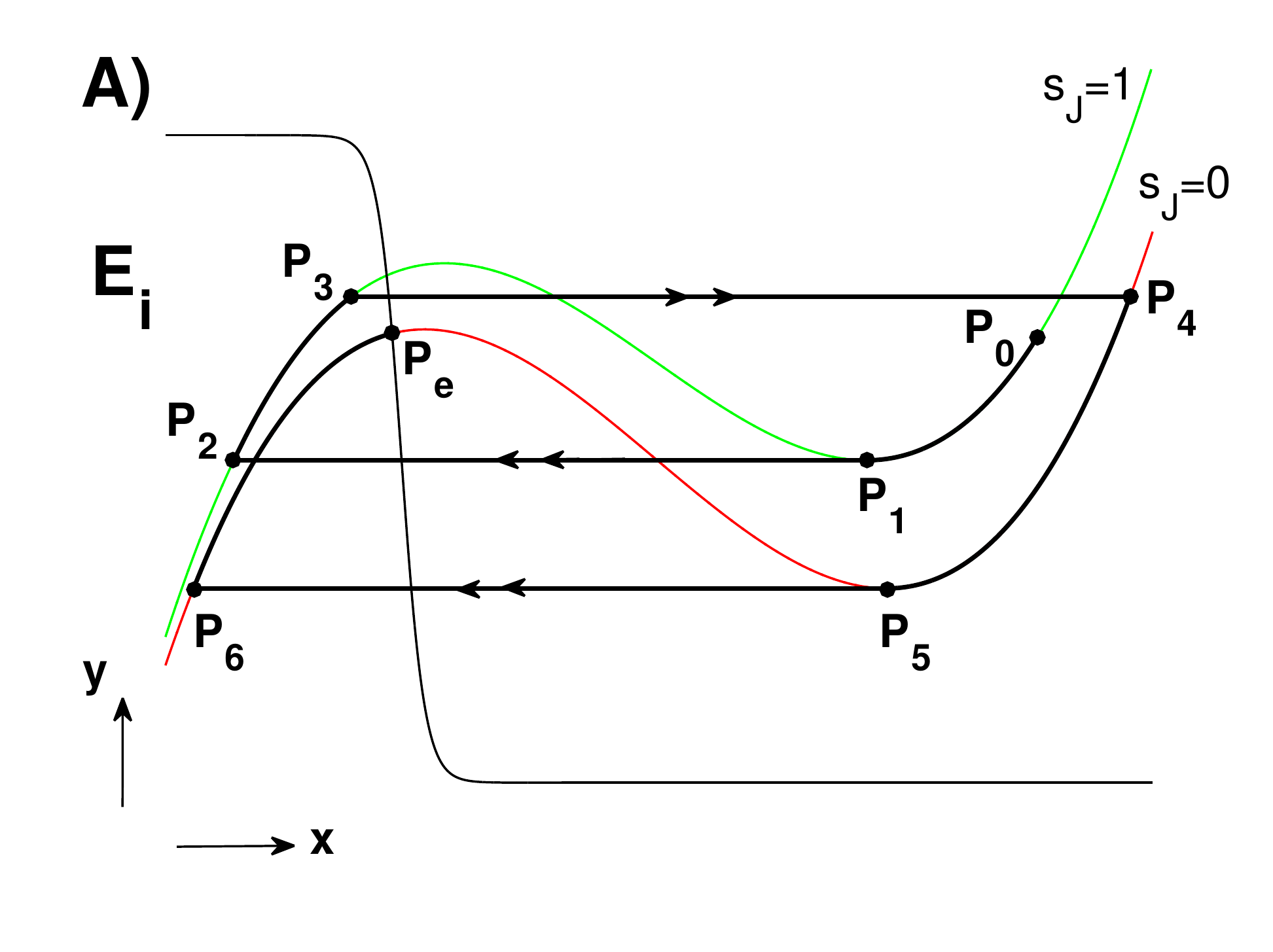}
  \includegraphics[height=52mm, width=58mm]{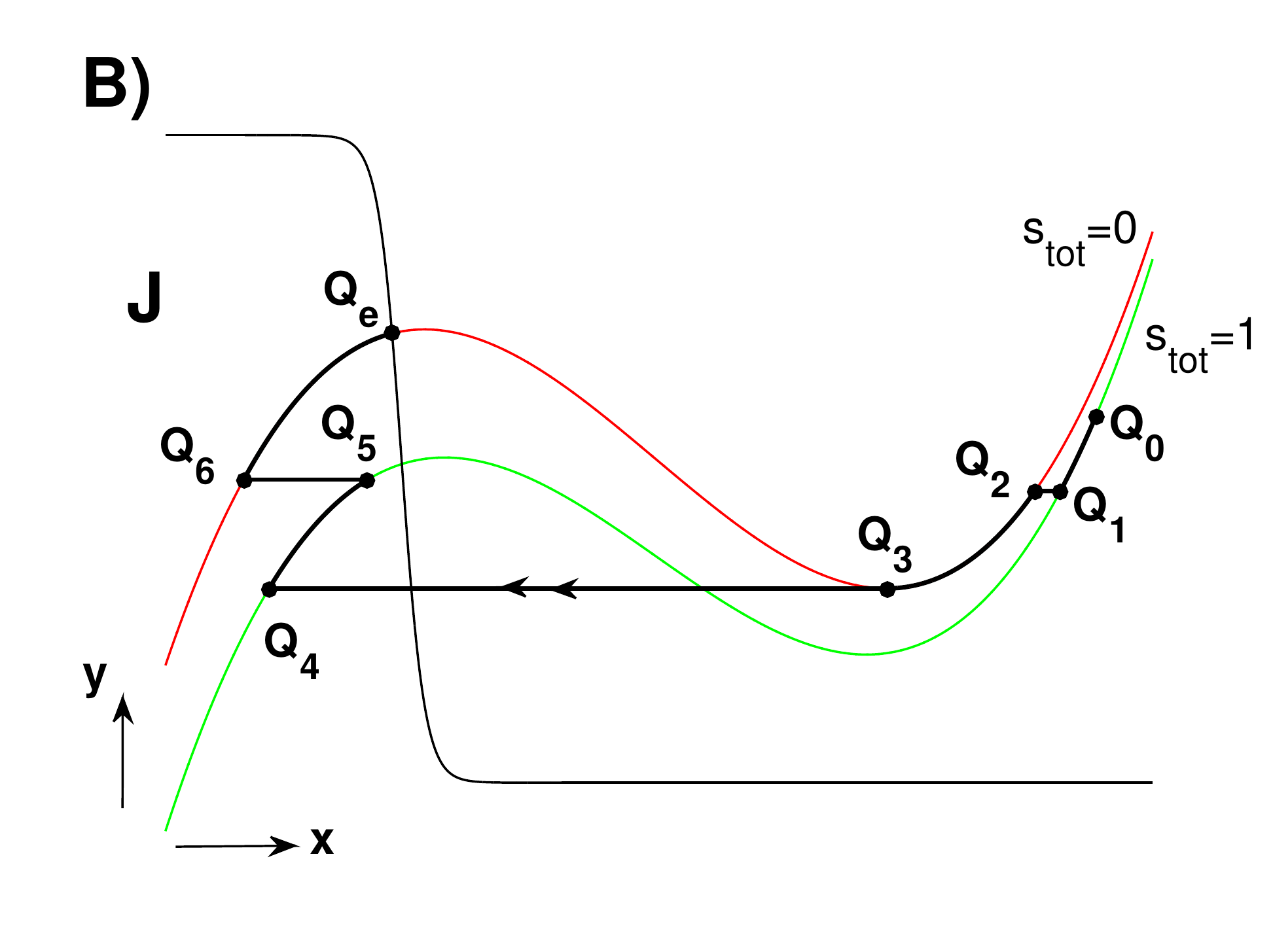}
  \caption{Plots of solution trajectories for A) $E$-cells and B) $J$-cell in black solid lines, approaching to their respective equilibrium points, if two populations start in the active phase and there is no delay in the synapses. These plots are for the case where $E$-cells lie above the left knee of the $s_J=0$ cubic, corresponding to the red curve of A), when $J$ jumps down.} 
  \label{nullclines_nodelay1}
\end{figure}

The previous situations occurred when the $E$-cells jumped down to the silent phase first. Now we consider the situation when the the $J$-cell jumps down first. The two cells start on the right branches of the $s_J=1$ and $s_{tot}=1$ nullclines, respectively, as in the previous cases. Here, however, the $J$-cell reaches the right knee of the $s_{tot}=1$ nullcline 
and then jumps down to the left branch of the $s_{tot}=1$ nullcline. This causes the $E$-cells to jump from the right branch of the $s_J=1$ nullcline to the right branch of the $s_J=0$ nullcline. The two populations travel on their respective nullclines until the $E$-cells reach the right knee and jumps to the left branch of the $s_J=0$ nullcline. This causes the $J$-cell to jump to the left branch of the $s_{tot}=0$ nullcline. Once again we are in Case 1.

{\bf Case 4:} The final case is where the the $J$-cell starts in the active phase and the
$E$-cells in the silent phase. The starting points are on the right branch of the $s_{tot}=0$ nullcline and on the left branch of the $s_J=1$ nullcline, respectively. The two populations travel on their respective nullclines until the $J$-cell reaches the right knee of the $s_{tot}=0$ nullcline and jumps down to the silent phase. Let $P_1$ be the position of the $E$-cells when the $J$-cell jumps down. There are two possibilities. If $P_1$ lies below the left
knee of the $s_J=0$ nullcline, then as the $J$-cell jumps down to the silent phase, 
the $E$-cells will jump to the left branch of the $s_J=0$ nullcline. This means that 
the $E$-cells will remain in the silent phase, so the $J$-cell will jump down to the left branch 
of the $s_{tot}=0$ nullcline. That is, both cells will be in the silent phase and we are 
in Case 1. The second possibility is that $P_1$ lies above the left knee of the $s_J=0$ nullcline.
In this case, as the $J$-cell jumps down to the silent phase the $E$-cells will jump up to the
active phase. Thus the $J$-cell will jump to a point on the left branch of the $s_{tot}=1$ nullcline and the $E$-cells to a point on the right branch of the $s_J=0$ nullcline. This means we are now in Case 2.
\end{proof}

\subsection {Longer active phase for $J$-cell}\label{subsec_longerJ}
In this section, we consider the case where $J$-cell has a sufficiently long active phase compared to the $E$-cells, and prove the existence and stability for synchronous periodic solutions when delays are present. 

\begin{theorem}
\label{theorem_same_initial}
A singular synchronous periodic solution exists with non-zero $\tau_J$ and $\tau_E$ if \\
(i) $y_F(1)>y_L(0)$, \\
(ii) the active phase of the $J$-cell is sufficiently long,\\
(iii) the delay $\tau_J$ is sufficiently large, and\\
(iv) the populations have overlapping active phases.
\end{theorem}

\begin{remark}
As noted in the proof of \cite{RUBIN02}, the condition $y_F(1)>y_L(0)$ indicates that the fixed point of the left branch of $C_1$ lies above the left knee of $C_0$. This makes it possible for $E$-cells to fire when they are released from inhibition.
\end{remark}

\begin{proof}
We prove the existence of a singular synchronous solution by constructing such a solution if the hypotheses of Theorem~\ref{theorem_same_initial} are satisfied. The number of $E$ oscillators in the network is set to two, but this construction easily generalizes to an arbitrary number. We assume the positions of the two $E$-cells are identical throughout the construction. The singular trajectory is shown in Figure~\ref{nullclines}.

\begin{figure} [htb!]
\centering
  \includegraphics[height=52mm, width=58mm]{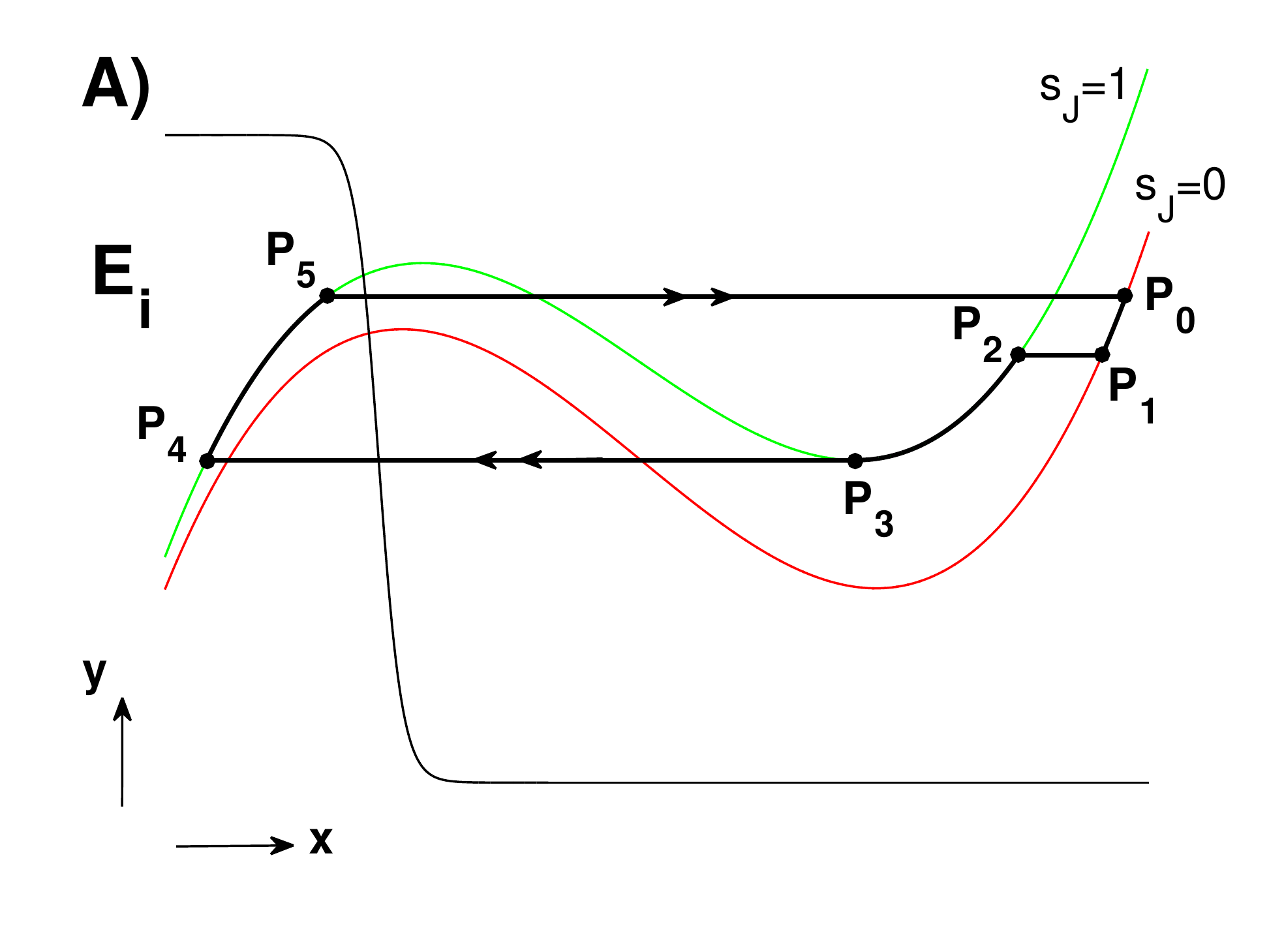}
  \includegraphics[height=52mm, width=58mm]{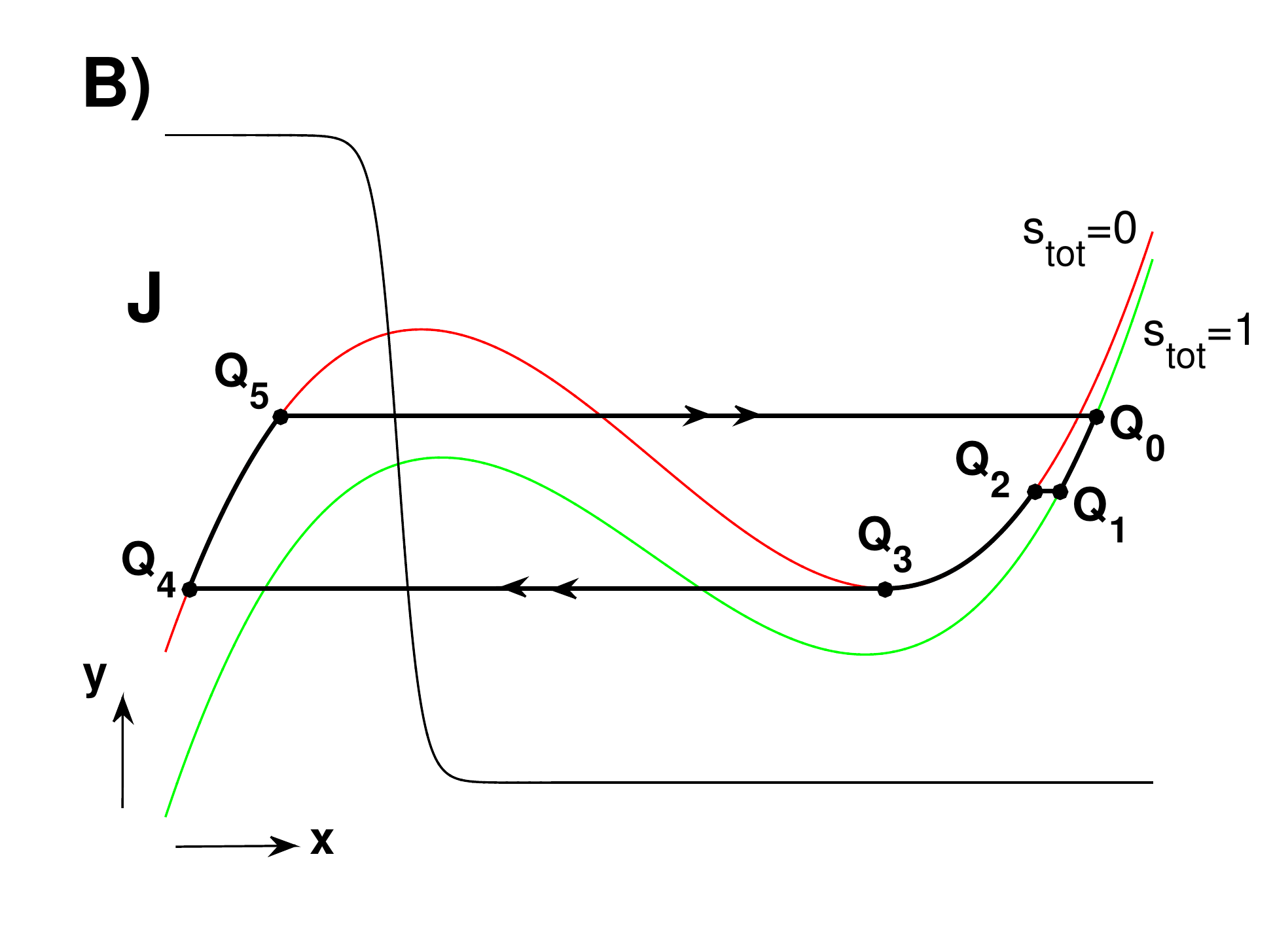}
  \caption{Plots of trajectories for A) $E$-cells and B) $J$-cell in black solid lines, if the two populations have overlapping active phases and there are delays in both synapses. Points $P_i$ and $Q_i$ correspond to the singular synchronous solution constructed in the text.}
\label{nullclines}
\end{figure}

We begin with the $E$-cells having just jumped to the active phase on the right branch of the $s_J=0$ cubic, labeled $P_0$ in Fig.~\ref{nullclines}A. Due to the excitation from the $E$-cells, the $J$-cell will jump to the right branch of the $s_{tot}=1$ cubic, labeled $Q_0$ in Fig.~\ref{nullclines}B, but only after an amount of time $\tau_E$. By assumption (iv), the $E$-cells are still in the active phase when the $J$-cell jumps up. 
Since $\tau_J$ is sufficiently large, there are two possible cases for the location of the $E$-cells depending on their relative position to the right knee of the $s_J=1$ cubic, when inhibition effectively turns on: $E$-cells lie above or below the right knee. The first possible trajectory for $E$-cells is given in Fig.~\ref{nullclines}A. When inhibition turns on the $E$-cells jump from $P_1$ to $P_2$ on the adjacent right branch of the $s_J=1$ cubic, while $J$ evolves down the right branch of the $s_{tot}=1$ cubic. If the $E$-cells lie below the right knee, on the other hand, turning on of the inhibition makes $E$ immediately jump down to the left branch of the $s_J=1$ branch. Regardless of the location of $E$-cells, however, both cases lead to the same result, that is, $E$-cells jump down before $J$-cell does (see below).

We assume that the $J$-cell has a longer active phase than the $E$-cells. This condition implies that, if $E$-cells still travel the right branch of the $s_J=0$ cubic corresponding to the first case mentioned above, the $E$-cells reach the right knee $P_3$ in Fig.~\ref{nullclines}A before the $J$-cell reaches the right knee of the $s_{tot}=1$ curve. Thus, at the time when $E$ jumps down to the left branch, labeled $P_4$ on the left branch of the $s_J=1$ cubic, $J$ lies above the right knee of the $s_{tot}=0$ cubic. Due to the time delay in the excitatory synapses, the turn-off of excitation to the $J$-cell does not follow immediately. 

Depending on the size of time delay $\tau_E$ compared to the remaining time for $J$ in the active phase, we have the following three cases, whose trajectories for the $J$-cell are given in Fig.~\ref{nullclines}B and Fig.~\ref{nullclines_J} : i) If the $J$-cell still lies above the right knee of the $s_{tot}=0$ curve after $\tau_E$ time, labeled $Q_1$ in Fig.~\ref{nullclines}B, the $J$-cell first jumps to the point $Q_2$ along the $s_{tot}=0$ cubic. Then, the $J$-cell moves down the right branch of $s_{tot}=0$ cubic while $E_i$ moves up the left branch of the $s_J=1$ cubic. When the $J$-cell reaches the right knee $Q_3$, it jumps down to the point $Q_4$ on the left branch of the $s_{tot}=0$ cubic; ii) If the $J$-cell lies below the right knee of the $s_{tot}=0$ cubic, $\tau_E$ time after the $E_i$ jumps down to $P_4$, the $J$-cell would immediately jump down to the left branch of the $s_{tot}=0$ cubic, as shown in Fig.~\ref{nullclines_J}A; iii) If the $J$-cell reaches the right knee of the $s_{tot}=1$ cubic within the $\tau_E$ time window, after the $E_i$ jump down, the $J$-cell first jumps down to the left branch of $s_{tot}=1$ and then moves upwards until the turn-off of excitation becomes effective. When this happens, the $J$-cell jumps to the left branch of the $s_{tot}=0$ cubic and starts to move upwards, which is shown in Fig.~\ref{nullclines_J}B. Note that, regardless of the location of the $J$-cell along the right branch of the $s_{tot}=1$ cubic after the $\tau_E$ time, $J$-cell eventually reaches the left branch of the $s_{tot}=0$ cubic while $E_i$ travels upwards along the left branch of the $s_J=1$ cubic. 

\begin{figure} [htb!]
\centering
\includegraphics[height=52mm, width=58mm]{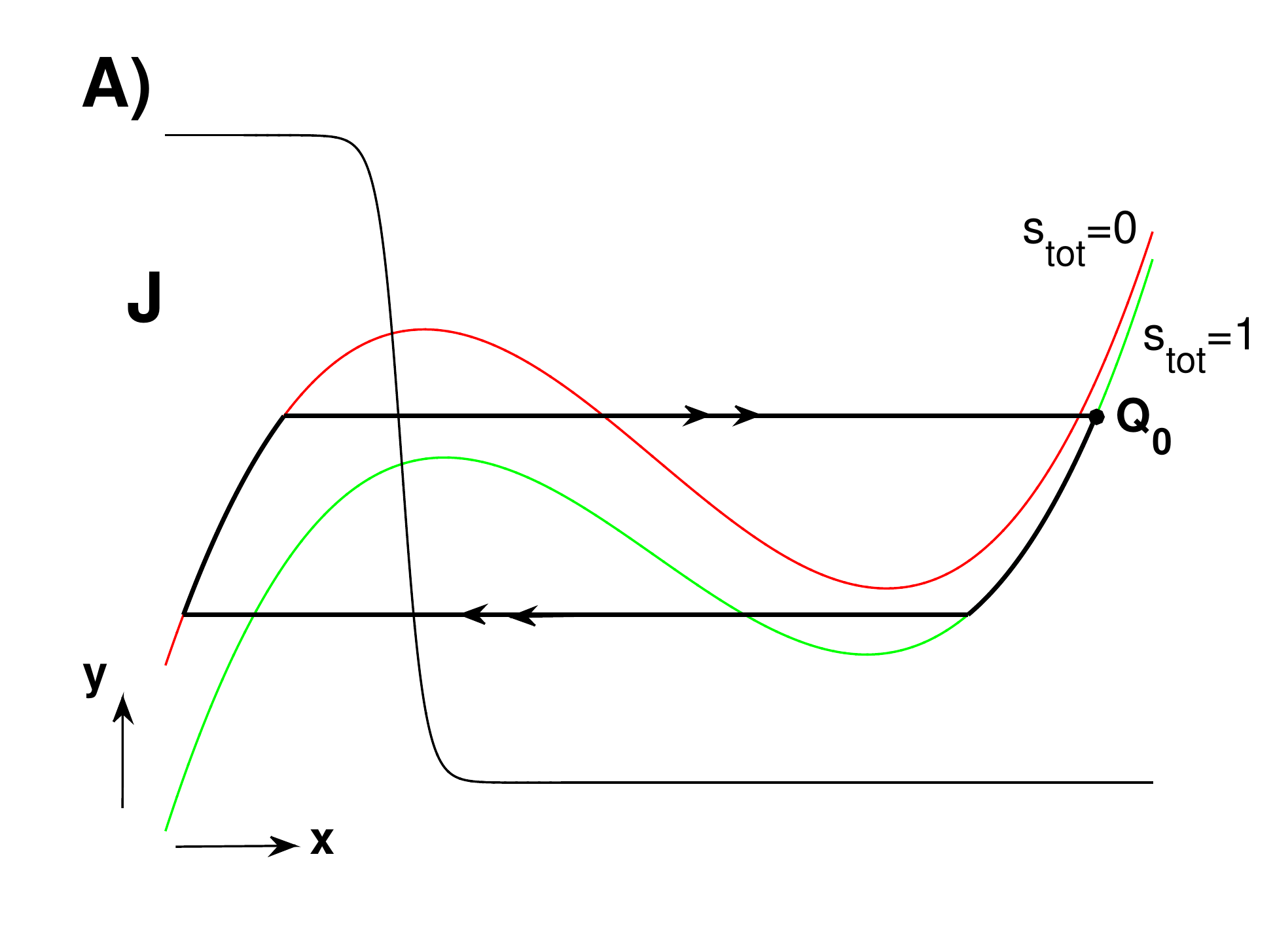}
  \includegraphics[height=52mm, width=58mm]{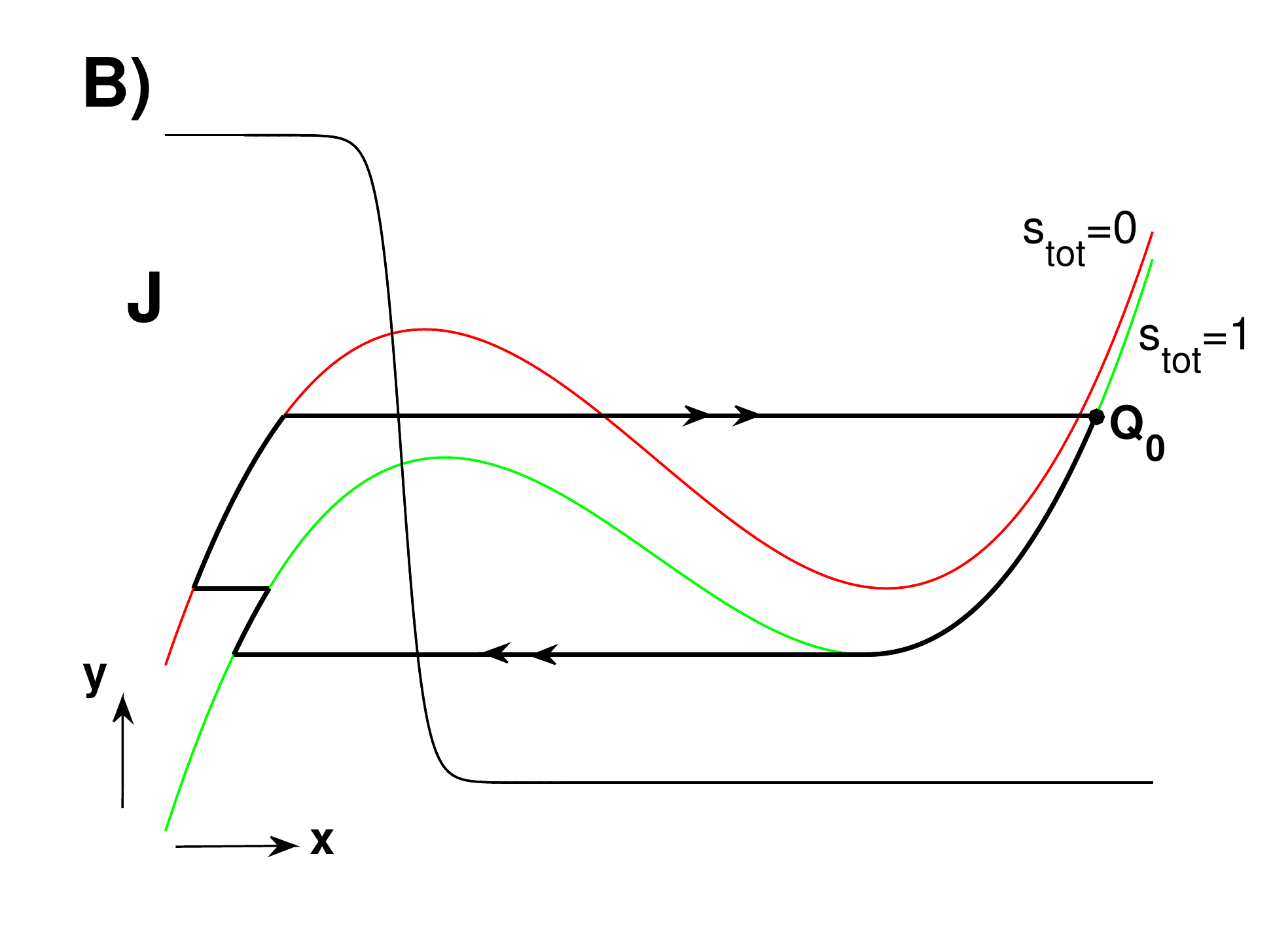}
  \caption{Plots of two different trajectories in $x$-$y$ phase plane, depending on the $J$-cell's position along the right branch of the $s_{tot}=1$ cubic when excitation to the $J$-cell turns off.}
  \label{nullclines_J}
\end{figure}

Now the inhibition to $E$-cells starts to turn off. However, due to the time delay $\tau_J$ in inhibition, this turn-off will not occur immediately. This means that the $E$-cells do not jump to the left branch of the $s_J=0$ cubic but continue to move up that of the $s_J=1$ cubic instead. Moreover, because of sufficiently large inhibitory delay $\tau_J$, the $E$-cells are able to reach the point above the left knee of the $s_J=0$ cubic, and to jump up to the right branch of the $s_J=0$ cubic again when they are finally released from inhibition. After the $E$-cells jump up, there is a delay of $\tau_E$ until turn-on of excitation becomes effective. However, even if $\tau_E=0$, due to the sufficiently large $\tau_J$, the $J$-cell lies above the left knee of the $s_{tot}=1$ cubic when excitation turns on. Therefore, it also jumps to the right branch of the $s_{tot}=1$ cubic and returns to its starting point, $Q_0$.

Existence of the synchronous solution requires that the $E$-cells can reach the point above the left knee of the $s_J=0$ cubic and escape from the silent phase when inhibition from the $J$-cell turns off, despite the fact that each cell is excitatory. Rubin and Terman in~\cite{RUBIN02} showed that this is indeed the case if there exists an additional slow variable which governs the rate of synaptic coupling term, $\dot s_J$. However, our model does not include an additional equation corresponding to this rate; the synaptic effect is modeled as a direct function of the corresponding $x$-variables. We demonstrate that the existence of the time delay in the synapses replaces this additional slow variable needed for the existence of the synchronized oscillations in $E$-cells.

As Rubin and Terman showed the lower bound on the duration of the $J$-cell active phase such that the synchronous solution can exist, we give similar estimates on this bound. The singular synchronous solution exists if $E$-cells lie in the region where $y_i>y_L(0)$ when the inhibition turns off after the $\tau_J$ time. This will follow if the active phase of the $J$-cell is sufficiently long and $\tau_J$ is sufficiently large.

Recall that $y_L(s_J)$ (or $y_R(s_J)$) is the $y$-value of the left (or right) knee of the $s=s_J$ cubic. We have $y_R(0)<y_L(0)$. Let $\tau_{esc}$ denote the time for $y$ to increase from $y_R(0)$ to $y_L(0)$ under $y'=G_L(y,1)$ on the slow subsystem. Let $\tau_a^J$ denote the duration of $J$-cell active phase. 
The synchronous oscillatory solution will exist if $\tau_{esc}$ is less that the time between
when the $E$-cells jump down to the silent phase and when the $E$-cells are released
from inhibition, i.e., the duration of the $E$-cell silent phase, $\tau_s^E$. Even though there are three possible cases for the position of the $J$-cell after $E$-cells jump down as identified above, we show that all these cases result in the same lower bound on $\tau_a^J$.

{\bf Case I:} The $J$-cell still lies above the right knee of the $s_{tot}=0$ cubic after the $E$-cells jump down and the delay time $\tau_E$ (corresponding to Fig.~\ref{nullclines}B). 
Now $\tau_a^J$ consists of up to four different parts: (i) the first part corresponds to the time delay $\tau_J$ after $J$ has jumped up, (ii) the next part corresponds to time, denoted by $\tau_E^*$, needed for the $E$-cell to reach the right knee of the $s_J=1$ cubic after jumping from the adjacent $s_J=0$ cubic, (iii) after $E$ jumps down, $J$ spends the $\tau_E$ time before jumping to the right branch of the $s_{tot}=0$ cubic, and (iv) finally the time, denoted by $\tau_J^*$, needed for $J$ to reach the right knee of the $s_{tot}=0$ cubic. Combining all times yields $\tau_a^J=\tau_J+\tau_E^*+\tau_E+\tau_J^*$. Similarly, $\tau_s^E$ consists of three parts: (i) $\tau_E$, the time after the $E$-cells
jump down to when the $J$-cell feel excitation turn off, (ii) $\tau_J^*$ as above, and
(iii) the delay $\tau_J$ after the $J$-cell jumps down before the $E$ cell feels the turn-off of the
inhibition. Thus, the singular solution exists if $\tau_{esc}<\tau_E+\tau_J^*+\tau_J$. Replacing the right hand side in the inequality using the equation for $\tau_a^J$ and rearranging the terms, it follows that $\tau_{esc}+\tau_E^*<\tau_a^J$.

{\bf Case II:} The $J$-cell lies below the right knee of the $s_{tot}=0$ cubic when excitation turns off (corresponding to Fig.~\ref{nullclines_J}A). In this case, the $J$-cell immediately jumps down to the silent phase. This implies that $\tau_J^*=0$ in the expressions from Case I, but  we have the same lower bound: $\tau_{esc}+\tau_E^*<\tau_a^J$.

{\bf Case III:} If the $J$-cell reaches the right knee of the $s_{tot}=1$ cubic during the $\tau_E$ delay time (corresponding to Fig.~\ref{nullclines_J}B), we can break down $\tau_E$ into two parts: $\tau_{E1}$ corresponding to time needed for $J$ to reach the right knee and $\tau_{E2}$ corresponding to time needed for $J$ to travel along the left branch of $s_{tot}=1$ cubic after jumping down. Then, the $E$-cells spend time $\tau_{E1}+\tau_J$ in the silent phase as $\tau_{E2}$ is a part of $\tau_J$.  Note that $\tau_a^J=\tau_J+\tau_E^*+\tau_{E1}$. Thus, it implies the same lower bound on $\tau_a^J$ as in both cases above.
 \end{proof}
 
\paragraph{Stability} To show the stability of the synchronous solution, suppose we perturb the system such that the starting positions for the two $E$-cells are slightly different in the active phase. Since the $y$-nullcline is a monotone decreasing curve, which intersects the $x$-nullcline at its left branch as shown in Fig.~\ref{nullclines}A, the function $g$ in Eq.~(\ref{excy}) is negative above the $y$-nullcline and its magnitude increases as $y_i$ moves up the right branch of the $x$-nullcline. This implies that the speed of the $y_i$-variable decreases as $E_i$ moves down along the right branch and becomes closer to the $y$-nullcline. Thus, because the $E$ cell behind moves down faster, the distance between the two cells decreases as they move down in the active phase. 

If the active phase for the $J$-cell is sufficiently long and the delay, $\tau_J$, is sufficiently large, both $E$-cells jump down to the left branch before the $J$ cell. If both cells jump at the same time (due inhibition turning on), their relative positions will be reversed (the one behind becoming the one in front) but the distance between them will be preserved. If the cell in front reaches the right knee of the $s_J=0$ nullcline while inhibition is still off and jumps to the left branch of the $s_J=0$ nullcline, it will move more rapidly than the lagging cell and the distance between them may increase.
However, once the second cell jumps, either due to reaching the right knee of the $s_J=0$ nullcline or
due to inhibition turning on, both cells will be on the same branch. In any case, both will end
up on the left branch of the $s_J=1$ nullcline after inhibition turns on. As on the right branch the velocity in the $y$ direction decreases as the $E$-cells move up along the left branch of the $s_J=1$ nullcline. If the active phase for the $J$-cell is sufficiently long and the delay, $\tau_J$, is sufficiently large, both the $E$-cells are able to reach the point above the left knee of the $s_J=0$ cubic, when inhibition turns off, and thus will jump up to the active phase simultaneously.  If we take the perturbation between the cells sufficiently small, the compression of the trajectories due to the difference in the velocity along either branch, can compensate for any expansion during the jumps.

As Rubin and Terman noted in \cite{RUBIN02}, the domain of attraction of the synchronous solution depends on the ability of the $E$-cells to pass through the ``window of opportunity," which is provided by the delay in our model. The size of this domain grows as either the $J$-cell's active phase or the size of $\tau_J$ increases.

\begin{remark}
The analysis leads to simple formulas for the period of the synchronous solution. Let $\tau_a^J$ be, as defined above, the time for $J$-cell to spend in the active phase. After $J$-cell jumps down to the silent phase, it first spends $\tau_J$ time while $E$ moves up the left branch of the $s_J=1$ cubic to reach the point above the left knee of the $s_J=0$ cubic. Then the $J$-cell spends another $\tau_E$ time until the $J$-cell receives excitation and escapes from the silent phase. It follows that the period of the synchronous solution is simply $\tau_a^J+\tau_J+\tau_E$. $\tau_a^J$ and $\tau_J$ are determined by the dynamics of the $J$-cell, while $\tau_E$ is determined by those of the $E$-cells. 
\end{remark}

\begin{remark}
\label{remark_same_initial}
As noted in the proof of Theorem~\ref{theorem_same_initial} above, the presence of the delay $\tau_J$ in inhibition is a sufficient condition for the existence of the synchronized solution, if we assume that the two populations have overlapping active phases. However, we could omit the delay $\tau_E$ and obtain the same results. Suppose that there is no delay from $E$-cells (i.e.,~$\tau_E=0$). Then the $J$-cell lies above the right knee of the $s_{tot}=0$ cubic when the $E$-cells jump down to the left branch of the $s_J=1$ cubic, and the rest of trajectories will follow the first case where $J$ jumps from $Q_1$ to $Q_2$ and then jumps down to $Q_4$ as soon as reaching the right knee, $Q_3$, as illustrated in Fig.~\ref{nullclines}B.  
\end{remark}

Based on our analysis, one may expect that it is possible that the synchronous solution could exist in the case where there is no delay in inhibition (i.e.,~$\tau_J=0$) but there is nonzero delay in excitatory synapses. This is true, but this case requires a different condition for the relative phases of the $E$ and $J$ populations. In Theorem~\ref{theorem_same_initial}, we assume that the active phases of two populations overlap. In order for the synchronous solution to exist without delay in inhibition, we need to assume the opposite. That is, when the $E$-cells are in the active phase the $J$-cell are in the silent phase and vice versa. Then, the synchronous solution can exist in the absence of delay from the $J$-cell. The proof of this case is given below.

\begin{corollary}
\label{corollary_opposite}
Suppose that there is no delay from the $J$-cell (i.e.,~$\tau_J=0$) while having the sufficiently long delay $\tau_E$ from the $E$-cells. Then, a singular synchronous periodic solution among the $E$-cells exists if 
(i)--(ii) in Theorem~\ref{theorem_same_initial} are satisfied, and the two populations start in opposite phases.
\end{corollary}

\begin{proof}
We begin with the $J$-cell in the silent phase just after it has jumped down from the right to the left 
branch of the $s_{tot}=0$ cubic, corresponding to $Q_0$ shown in Fig.~\ref{nullclines_opp}B. The $E$-cells 
are sitting on the left branch of $s_J=1$ cubic at the point above the left knee of the $s_J=0$ cubic, 
labeled $P_3$ in Fig.~\ref{nullclines_opp}A.  Since there is no delay in the inhibitory synapses, 
the $E$-cells will jump to $P_0$ along the right branch of $s_J=0$ cubic as soon as the $J$-cell jumps down. 

\begin{figure} [htb!]
\centering
\includegraphics[height=52mm, width=58mm]{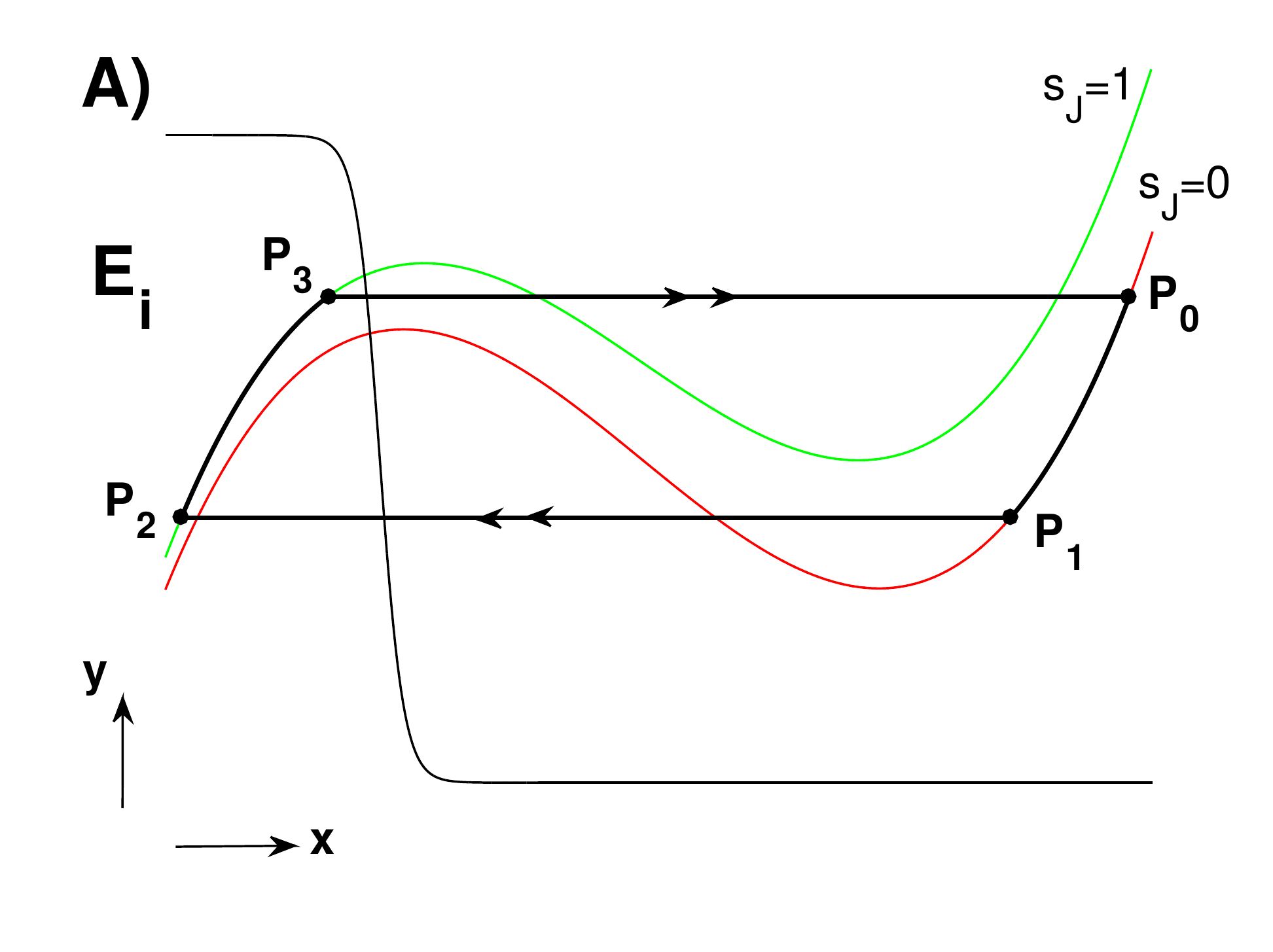}
  \includegraphics[height=52mm, width=58mm]{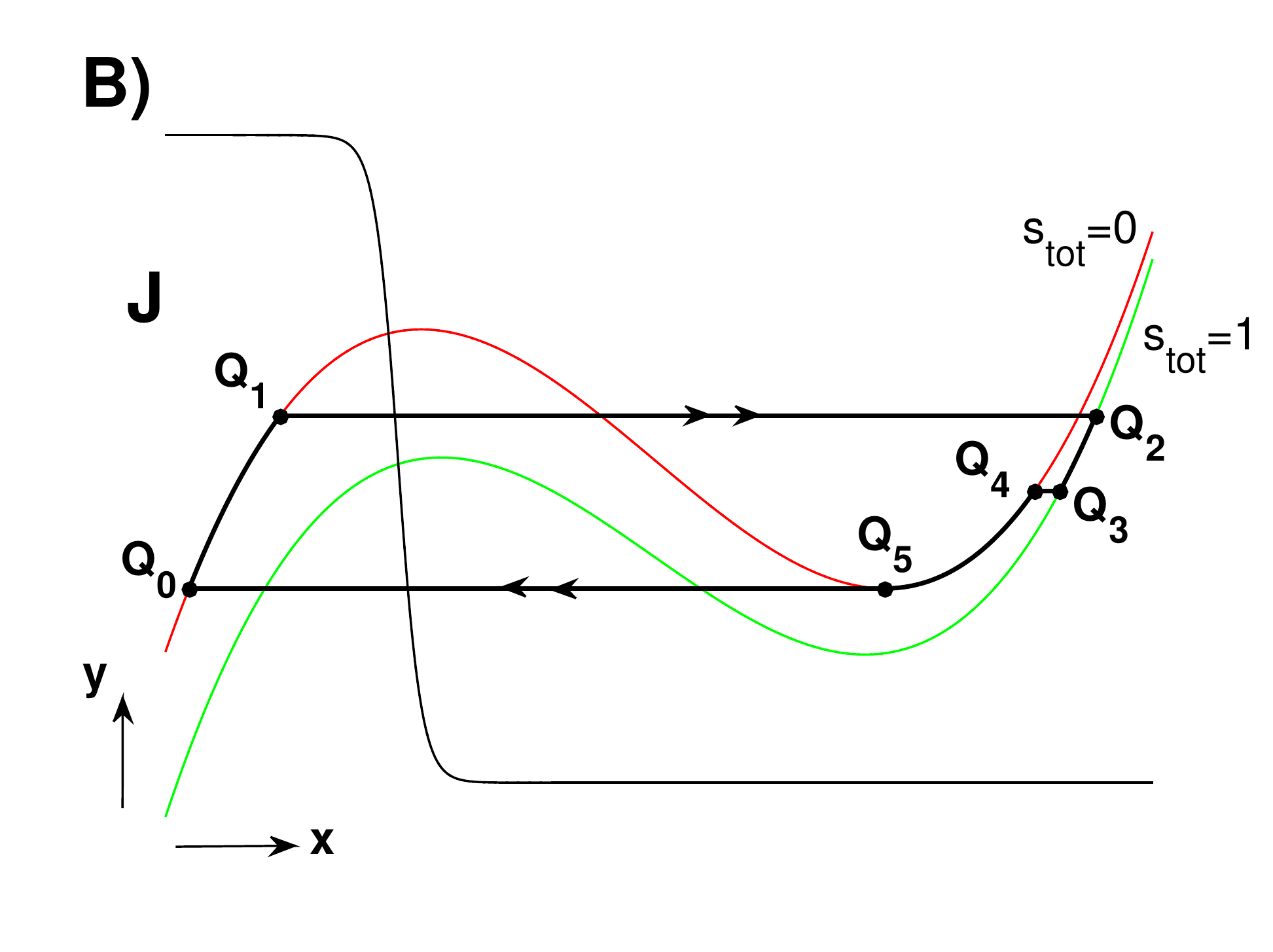}
  \caption{Plots of trajectories for A) $E$-cells and B) $J$-cell in black solid lines, if the two populations start in opposite phases and there is no delay in inhibition. Points $P_i$ and $Q_i$ correspond to the singular synchronous solution constructed in the text.}
  \label{nullclines_opp}
\end{figure}

Since the delay from the $E$-cells is sufficiently large, it enables the $J$-cell to reach the point $Q_1$ above the left knee of the $s_{tot}=1$ cubic, at the time when the excitation becomes effective. This makes it possible for the $J$-cell to jump up again to the active phase, labeled $Q_2$, while the $E$-cells move down along the right branch of the $s_J=0$ cubic. 
The sufficiently long delay $\tau_E$ in excitation also enables the $E$-cells to move down and reach the point, labeled $P_1$, below the right knee of the $s_J=1$ cubic when the $J$-cell jumps up. It follows that $E$ immediately jumps down to the left branch of the $s_J=1$ cubic labeled $P_2$ because of no inhibitory delay $\tau_J$. After the $\tau_E$ time, the $J$-cell jumps from $Q_3$ to the left branch of the adjacent cubic, $Q_4$. Since we assume that the $J$-cell has a sufficiently long active phase and $\tau_E$ is sufficiently large, the $E$-cells are able to reach the point $P_3$ above the left knee of the $s_J=0$ cubic before the $J$-cell reaches $Q_5$ and returns to its starting point, $Q_0$. This completes one cycle of the singular periodic synchronous solution.
\end{proof}

The analysis of the stability of the synchronous solution, if the $E$ and $J$ populations are in opposite phases, is similar to the previous case. Because of the difference in velocity along either branch, if one $E$-cell moves ahead the other cell eventually is caught up, and both cells are able to escape from the silent phase when inhibition turns off, if the $J$-cell active phase is sufficiently long and the sufficiently large delay in excitation is present.

\subsection {Longer active phase for $E$-cells} \label{subsec_longerE}
In the section above, one fundamental assumption is that the $J$-cell has a longer active phase than $E$-cells, which is motivated by biological observation in thalamocortical networks. However, Doiron et al.~\cite{doiron2003} have shown a synchronizing effect of global inhibitory feedback in a model which only represents the spiking dynamics of the excitatory cells. Thus, here we consider the case $E$-cell has a longer active phase. That is, we suppose that the $E$-cells and $J$-cell have overlapping active phases as described in Theorem~\ref{theorem_same_initial}. However, we now suppose that the $J$-cell reaches the right knee of the $s_{tot}=1$ cubic first and jumps down to its left branch before the $E$-cells. The stability for the synchronous solution can be shown in an analogous manner to Section~\ref{subsec_longerJ}.

\begin{theorem}
\label{theorem_longE}
A singular synchronous periodic solution exists if \\
(i) $y_F(1)>y_L(0)$, \\
(ii) the active phase of the $E$-cell is sufficiently long,\\
(iii) the delay $\tau_J$ is sufficiently large, and\\
(iv) the populations have overlapping active phases.
\end{theorem}

\begin{proof}
We begin with the $J$-cell in its silent phase and the $E$-cells having just jumped up to the active phase of the $s_J=0$ cubic. This is followed, when excitation turns on a time $\tau_E$ later, by the $J$-cell jumping up to the right branch of the $s_{tot}=1$ cubic. See Figure~\ref{nullclines_longE} for the solution trajectories of both populations. 

Since the active phase of the $E$-cell is sufficiently long, the $J$-cell reaches the right knee of the $s_{tot}=1$ cubic, $Q_1$, before the $E$-cell reaches its respective right knee. This makes the $J$-cell jump down to the left branch of the $s_{tot}=1$ cubic, labeled $Q_2$. Also, as $\tau_J$ is sufficiently large, it is possible for the $J$-cell to reach the right knee and jump down before the $E$-cells feel the effect of inhibition. Once the inhibition turns on, the $E$-cells jump down to the left branch of the $s_J=1$ cubic regardless of their positions along the right branch of the $s_J=0$ cubic. The sufficiently large $\tau_J$ ensures that the $E$-cells can reach the point slightly above or below the right knee of the $s_J=1$ cubic so that they jump down soon after inhibition effectively turns on. One possible trajectory for the $E$-cells, in which they lie below the right knee of the $s_J=1$ cubic, is illustrated in Fig.~\ref{nullclines_longE}A. The other possible trajectory, in which the $E$-cells lie above the right knee, is similar to that shown in Fig.~\ref{nullclines}A. Even though the inhibition to the $E$-cells becomes effective before the $J$-cell jumps down, the $E$-cells jump down to the left branch of the $s_J=1$ cubic either immediately or soon after the $J$-cell jumps down. 

\begin{figure} [htb!]
\centering
\includegraphics[height=52mm, width=58mm]{nullclines_E1-eps-converted-to.pdf}
  \includegraphics[height=52mm, width=58mm]{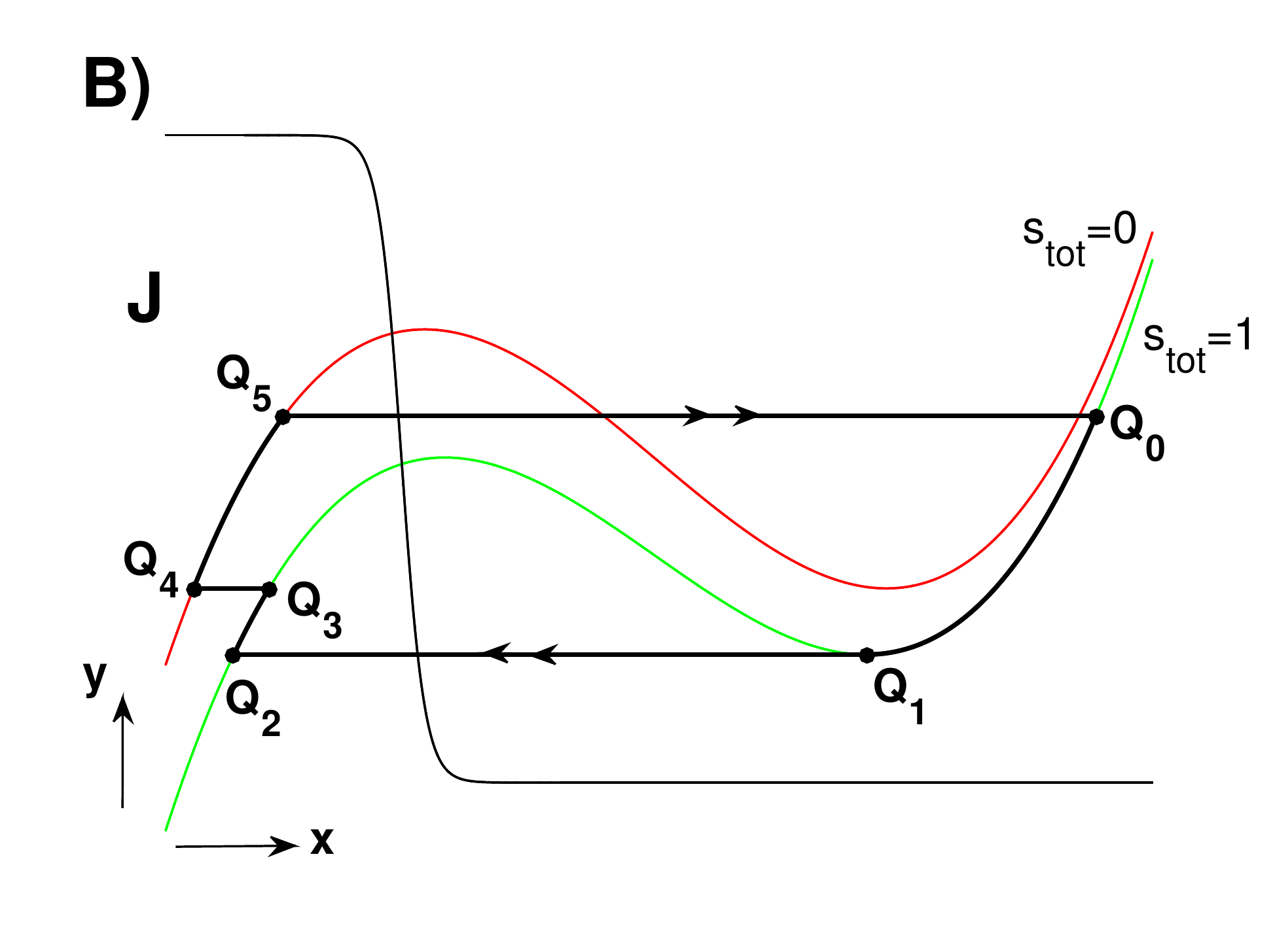}
  \caption{Plots of solution trajectories for A) $E$-cells and B) $J$-cell in black solid lines, if the two populations have overlapping active phases and the $E$-cells have longer active phase than the $J$-cell.}
  \label{nullclines_longE}
\end{figure}

The $\tau_E$ time later after the $E$-cells jump down, excitation to the $J$-cell turns off, which results in the $J$-cell jumping from $Q_3$ to $Q_4$ while the $E$-cells continue to move up in the silent phase. Even though the $J$-cell jumps down before the $E$-cells, the latter still receive inhibition from the former because the delay $\tau_J$ is sufficiently long. This allows the $E$-cells to reach the point above the left knee of the $s_J=0$ cubic, $P_3$, and to jump up to $P_0$ when they are finally released from the inhibition. Once the $E$-cells jump up to the right branch of the $s_J=0$ cubic again. After the $\tau_E$ time the $J$-cell reaches the point $Q_5$ above the left knee of the $s_{tot}=1$ cubic and jumps up to its starting point $Q_0$. Thus, one full cycle of the synchronous solution is complete.
\end{proof}

\begin{remark}
\label{remark_longE}
As noted in one of the hypotheses above, the delay $\tau_J$ being large enough is a sufficient condition for the existence of the synchronized solution, if we assume that the two populations have overlapping active phases. However, we could omit the delay $\tau_E$ in excitation and obtain the same results. 
All that is needed is that the $E$-cells have a sufficiently long active phase and the $\tau_J$ is sufficiently large. 
Suppose there is no delay from the $E$-cells (i.e.,~$\tau_E=0$), when the $E$-cells jump down to the left branch, excitation to the $J$-cell immediately turns off so that the $J$-cell jumps to the left branch of the $s_{tot}=0$ cubic. 
The rest of trajectories will follow Theorem~\ref{theorem_longE}, which implies that the $E$-cells move up the left branch of the $s_J=1$ cubic and jump again when they are released from inhibition after the large delay $\tau_J$. 
\end{remark}

Finally, we consider the case where there is no delay in inhibition while having delay in the excitatory synapses, i.e., $\tau_J=0$ and $\tau_E>0$. For a synchronous solution to exist, this case requires a different condition for the positions of both populations similar to Corollary~\ref{corollary_opposite}. We need both populations start in different phases, for example, the $E$-cells in the active and the $J$-cell in the silent phase. Then, the synchronous solution appears in the absence of delay from the $J$-cell. The proof of this case is given below.

\begin{corollary}
\label{corollary_longE_opposite}
Suppose that there is no delay inhibition (i.e.,~$\tau_J=0$) but there is a sufficiently large delay $\tau_E$ in excitation. Then, a singular synchronous periodic solution exists if (i)--(ii) in Theorem~\ref{theorem_longE} are satisfied, and the two populations start in opposite phases.
\end{corollary}

\begin{proof}
We begin with the starting positions for both populations as described in Corollary~\ref{corollary_opposite}: the $J$-cell is in the silent phase just after it has jumped down to the left branch of the $s_{tot}=0$ cubic, while the $E$-cells have just jumped up to the active phase due to the immediate release of inhibition. The reason why the $J$-cell jumps down to the $s_{tot}=0$ cubic is because the excitatory delay, $\tau_E$, is nonzero. 

Since $\tau_E$ is sufficiently large, the turn-on of excitation to the $J$-cell does not immediately follow when the $E$-cells jump up. In fact, this large delay enables the $J$-cell to reach the point above the left knee of the $s_{tot}=1$ cubic, at the time when excitation effectively turns on. When the $J$-cell jumps up to the right branch of the $s_{tot}=1$ cubic, the $E$-cells immediately receive inhibition. As shown in Corollary~\ref{corollary_opposite}, because of the large $\tau_E$ the $E$-cells lie below the right knee of the $s_J=1$ cubic when the $J$-cell jumps up. Thus, they immediately jump down to the left branch of the $s_J=1$ cubic. 

Now the situation repeats but in reverse. The $J$-cell moves down the right branch of $s_{tot}=1$ cubic and the $E$-cells move up the left branch of the $s_J=1$ cubic. Since $\tau_E$ is sufficiently large, when the excitation to the $J$-cell turns off, the $J$-cell will reach the point below the right knee of the $s_{tot}=0$ cubic and will jump down to the left branch of the $s_{tot}=0$ cubic. When this occurs the $E$-cells will be immediately released from inhibition. Since $\tau_E$ is sufficiently large, the $E$-cells will be above the left knee of the $s_J=0$ cubic and thus will jump up to the right branch of this cubic. Therefore, both populations return to their respective starting points, and one cycle of the synchronous solution is complete. 
\end{proof}

\begin{remark}
Note that the longer active phase for the $E$-cells is not the necessary hypothesis to ensure the existence of synchronous solution in Corollary~\ref{corollary_longE_opposite} above. In fact, in the proof of Corollary~\ref{corollary_opposite} the longer active phase for the $J$-cell is promoting the $E$-cells to escape from the silent phase but we could omit this condition from the hypotheses as long as $\tau_E$ is sufficiently large. Thus, Corollary~\ref{corollary_opposite}--\ref{corollary_longE_opposite} are essentially the same if we eliminate the condition on the active phase for each cell from their hypotheses.
\end{remark}

\section{Numerical Simulations} \label{numerical}
We conduct numerical simulations to illustrate the synchronous solutions under different conditions on the relative duration of the active phase for the $E$-cells and the $J$-cell, and on their delays, as constructed in Section~\ref{analysis}. We first consider a simple model by specifying explicit functions for $f$ and $g$ in Eqs.~(\ref{excx})--(\ref{excy}), for $f_J$ and $g_J$ in Eqs.~(\ref{inhx})--(\ref{inhy}), and synaptic variables $s_i$ and $s_J$. Then, we consider a more complex model for the thalamic spindle sleep rhythm, which resembles the globally inhibitory networks described in Section~\ref{model}.

 \subsection{Simple model}\label{simple}
The first model we consider is the following specific version of Eqs.~(\ref{excx})--(\ref{inhy})
 \begin{align} 
 \label{numeric_x}
\dot x_i&=3x_i-x_i^3+y_i-g_{inh}s_J(x_J(t-\tau_J))(x_i-x_{inh}),\\
\label{numeric_y}
\dot y_i&=\epsilon (\lambda -\gamma \tanh (\beta (x_i-\delta))-y_i),\\
\label{numeric_xJ}
\dot x_J&=3x_J-x_J^3+y_J-g_{exc}\left({1\over N}\sum_{i} s_i(x_i(t-\tau_E))\right)(x_J-x_{exc}),\\
\label{numeric_yJ}
\dot y_J&=\epsilon(\lambda_J -\gamma_J \tanh (\beta_J (x_J-\delta))-y_J),
\end{align}
Note that $f$ and $f_J$ are the same. Using different functions would not alter the results significantly. These functions are modified from those used in~\cite{terman95,CampWang98} so that the properties of $f$ and $g$ are as illustrated in Fig.~\ref{single_orbit}. The parameters $\beta,~\beta_J$ denote the steepness of the sigmoidal curves for the $y$-nullcline and $y_J$-nullcline, respectively, and we set both to be $ \gg 1$. The parameters $\lambda,~\gamma$, and $\delta$ for the $E$-cells are used to modify the amount of time  they spend in the left or right branches as their speed along either branch depends on the $y$-nullcline. Model parameters for the $J$-cell, which are different from those for $E$-cells, are similarly defined.

The coupling function, $s$, is defined to be a sigmoid curve having the form of
 \begin{align} 
 \label{coupling}
 s(x)=[1+\exp (-(x-\theta)/\sigma)]^{-1},
\end{align}
where $\sigma$ determines the steepness of this sigmoid and is set to be $\ll 1$. The parameter $\theta$ is the threshold for $x$-variable, i.e., the value at which $s$ rapidly changes from 0 to 1. 

\subsubsection{Longer active phase for the $J$-cell}

We used the simulation package XPPAUT~\cite{Ermentrout02} to numerically integrate Eqs.~(\ref{numeric_x})--(\ref{numeric_yJ}) and show the existence of stable synchronous solutions under the different cases we considered in Section~\ref{subsec_longerJ}. First, we consider the case where the $J$-cell has a longer active phase than the $E$-cells. Figure~\ref{simple_longJ} shows the synchronous solutions for two $E$-cells and one $J$-cell with different combinations of two coupling delays, $\tau_J$ and $\tau_E$; A) both delays are non-zero, specifically, $\tau_J=7$ and $\tau_E=3$,  B) the time delay in the excitatory synapse is zero, i.e.,~$\tau_E=0$, while $\tau_J=10$, and C) the time delay in the inhibitory synapse is zero, i.e.,~$\tau_J=0$, while $\tau_E=10$. Model parameters are chosen as described in Fig.~\ref{simple_longJ}. The black and green curves correspond to the oscillations in the $J$-cell's and one $E$-cell's voltages, respectively. The second $E$ cell is plotted in red, but the cells synchronize so quickly that this curve is not visible in the simulations.

\begin{figure} [htb!]
\centering
\includegraphics[height=55mm, width=58mm]{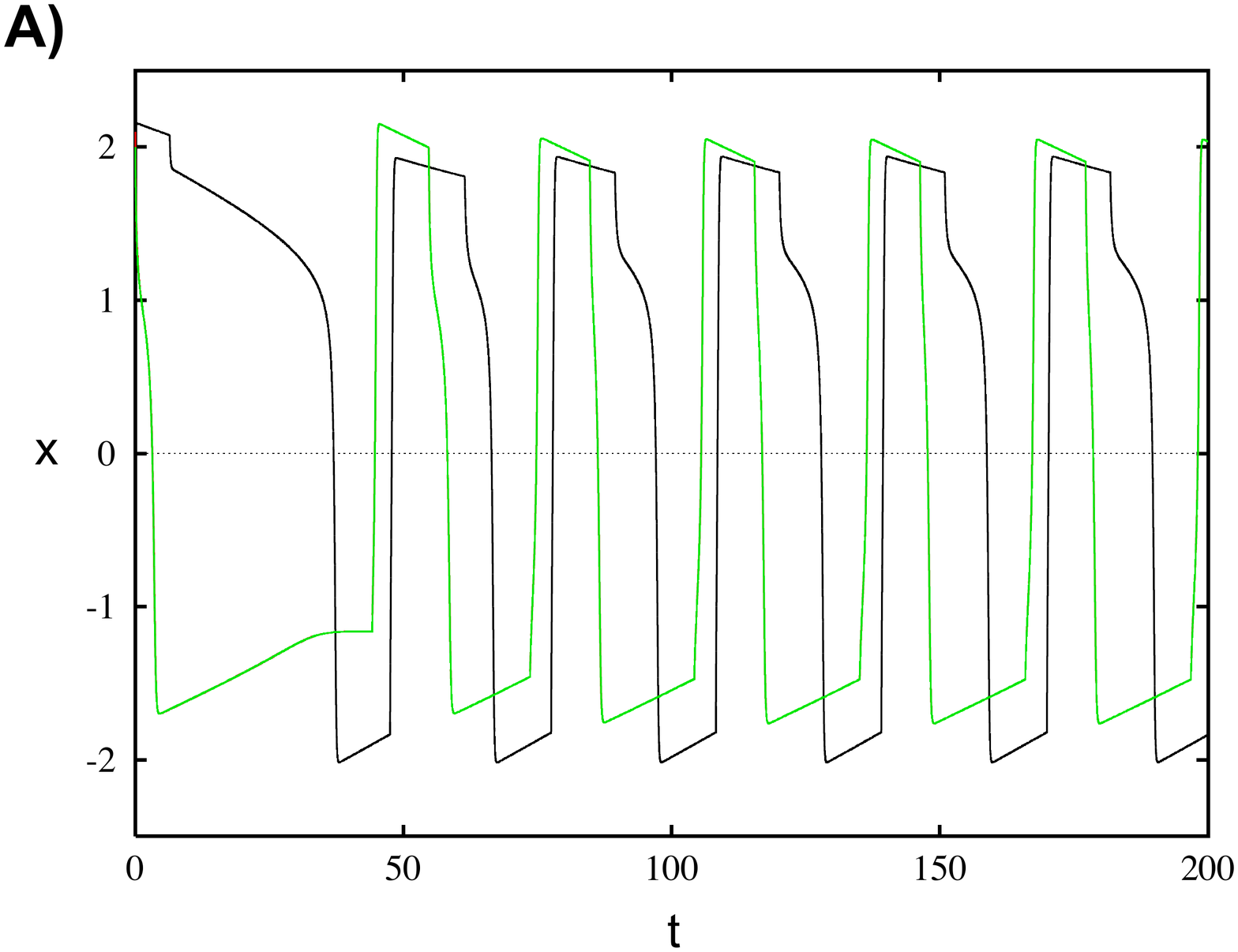}\\
\vspace{-0.5cm}
\includegraphics[height=55mm, width=58mm]{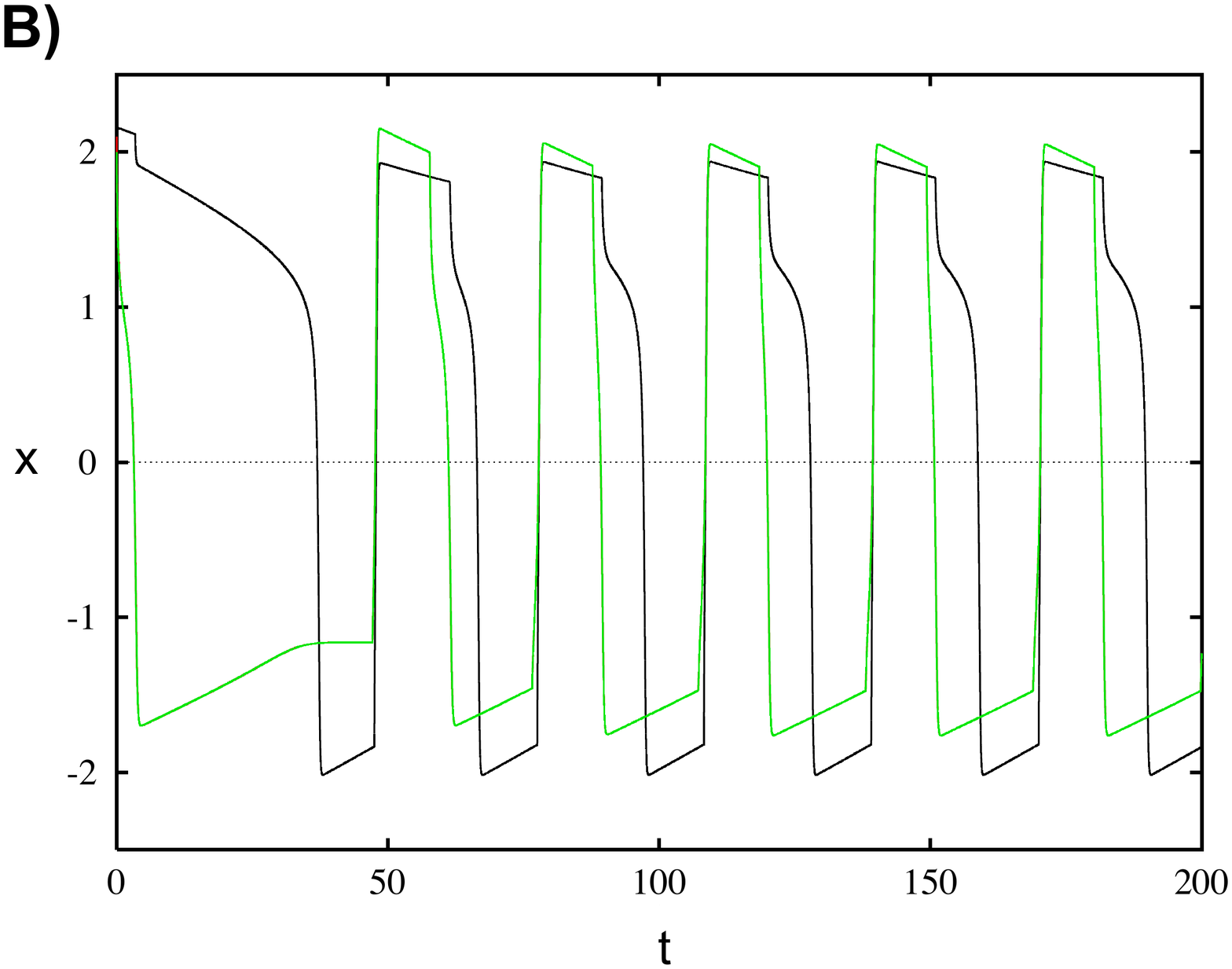}
\includegraphics[height=55mm, width=58mm]{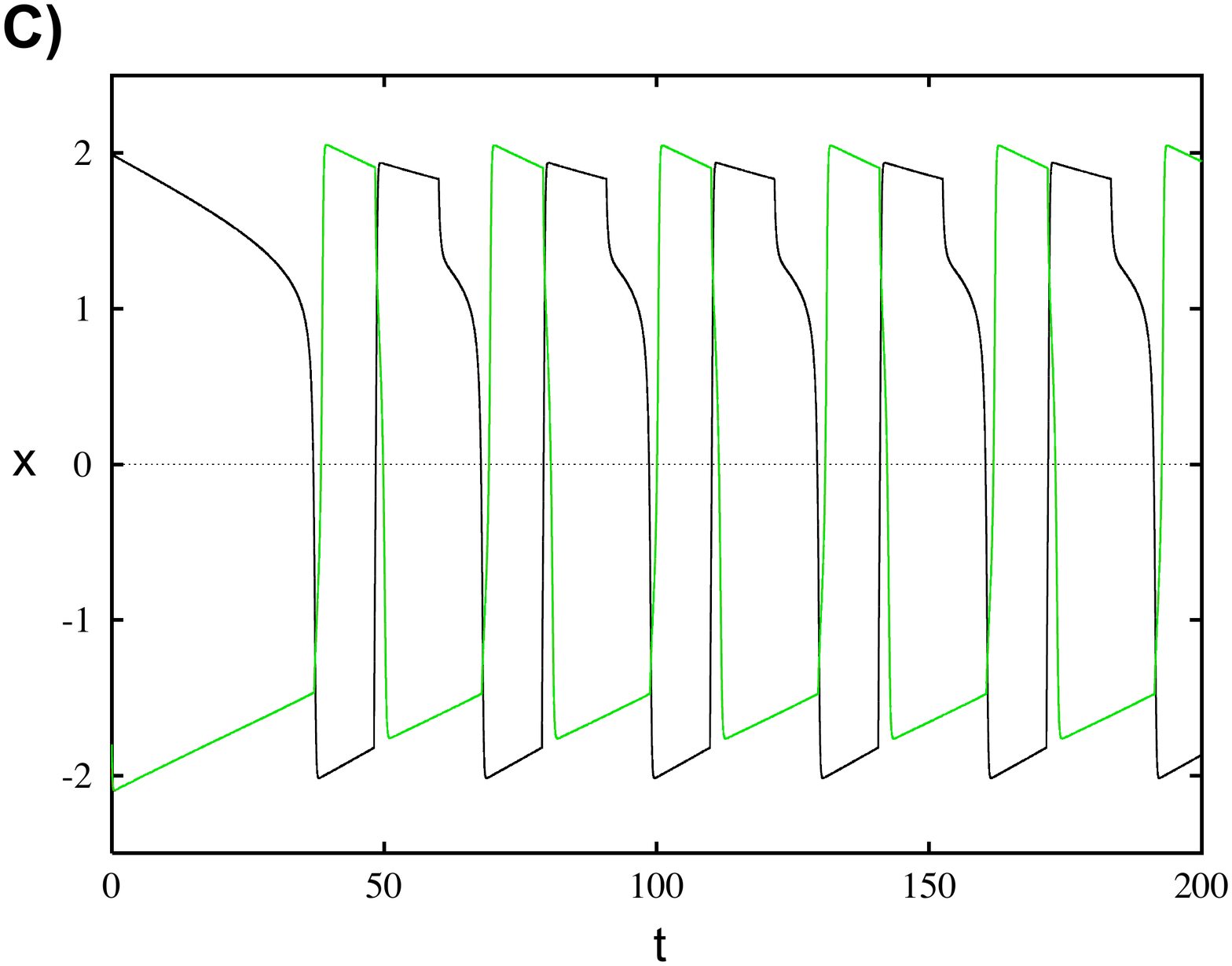}
\vspace{-0.5cm}
  \caption{Synchronous solutions for two $E$-cells and one $J$-cell with different combinations of the two coupling delays, when the $J$-cell has a longer active phase; A) $\tau_J=7$ and $\tau_E=3$, B) $\tau_J=10$ and $\tau_E=0$, and C) $\tau_J=0$ and $\tau_E=10$. The black and green/red curves are the time courses of $J$-cell and $E$-cell voltages, respectively. Parameter values used are $\epsilon=0.025,~\gamma=\gamma_J=5,~\beta=\beta_J=10,~\delta=\delta_J=-1.1,~\lambda=1,~\lambda_J=0,~\sigma=0.002$, and $\theta=-0.5$. Coupling parameter values are $g_{exc}=g_{inh}=1,~x_{inh}=-3$, and $x_{exc}=3$. }
  \label{simple_longJ}
\end{figure}

Note that in all three solutions of Fig.~\ref{simple_longJ}, the $J$-cell has a longer active phase than the $E$-cells. In the oscillations shown in Fig.~\ref{simple_longJ}A, specifically, the $J$-cell (black oscillations) fires $\tau_E$ time after the $E$-cells fire but the latter jump down before the former does, as described in Theorem~\ref{theorem_same_initial}. When the $E$-cells jump down, the $J$-cell is still in the active phase. However, the $J$-cell jumps to the right branch of another cubic after $\tau_E$ time and eventually jumps down to the silent phase as soon as it reaches the right knee of the same cubic (see Figs.~\ref{nullclines}A and B). By the time when inhibition from the $J$-cell completely turns off, the $E$-cells reach the point above the threshold for firing, thus they jump to the active phase again and complete one full cycle. In Fig.~\ref{simple_longJ}B, since there is no delay in the excitatory synapses (see Remark~\ref{remark_same_initial}), the $J$-cell fires immediately after the $E$-cells fire. The rest of trajectories are analogous to those of Fig.~\ref{simple_longJ}A. Finally, Fig.~\ref{simple_longJ}C shows the synchronous solution in the case that there is no delay in the inhibitory synapses. The $E$-cells fire when the $J$-cell jumps down, resulting in the two populations being in opposite phases, as described in Corollary~\ref{corollary_opposite}. 

\subsubsection{Longer active phase for the $E$-cells}
Now we consider the case where the $E$-cells have a longer active phase, described in Section~\ref{subsec_longerE}, by modifying relevant model parameters. Figure~\ref{simple_longE} shows the synchronous solutions for two $E$-cells and one $J$-cell with different combinations of two coupling delays, $\tau_J$ and $\tau_E$; A) both delays are non-zero, specifically, $\tau_J=30$ and $\tau_E=15$,  B) the time delay in the excitatory synapse is zero, i.e.,~$\tau_E=0$, while $\tau_J=45$, and C) the time delay in the inhibitory synapse is zero, i.e.,~$\tau_J=0$, while $\tau_E=45$. Model parameters are chosen as described in Fig.~\ref{simple_longE}. The black and green curves correspond to the oscillations in the $J$-cell's and one $E$-cell's voltages, respectively. The second $E$ cell is plotted in red, and not visible due to the synchronization.

\begin{figure} [htb!]
\centering
\includegraphics[height=55mm, width=58mm]{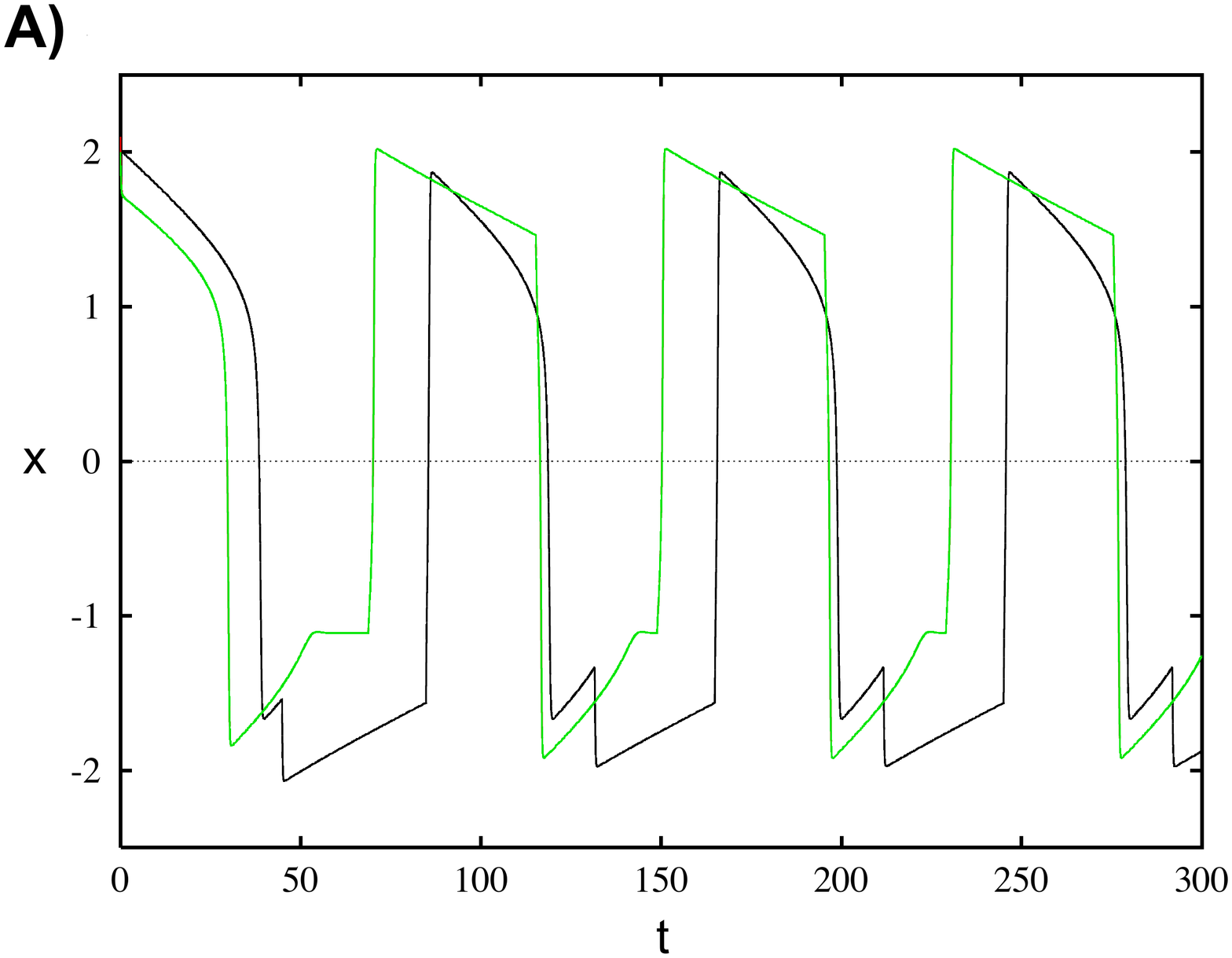}\\
\vspace{-0.5cm}
\includegraphics[height=55mm, width=58mm]{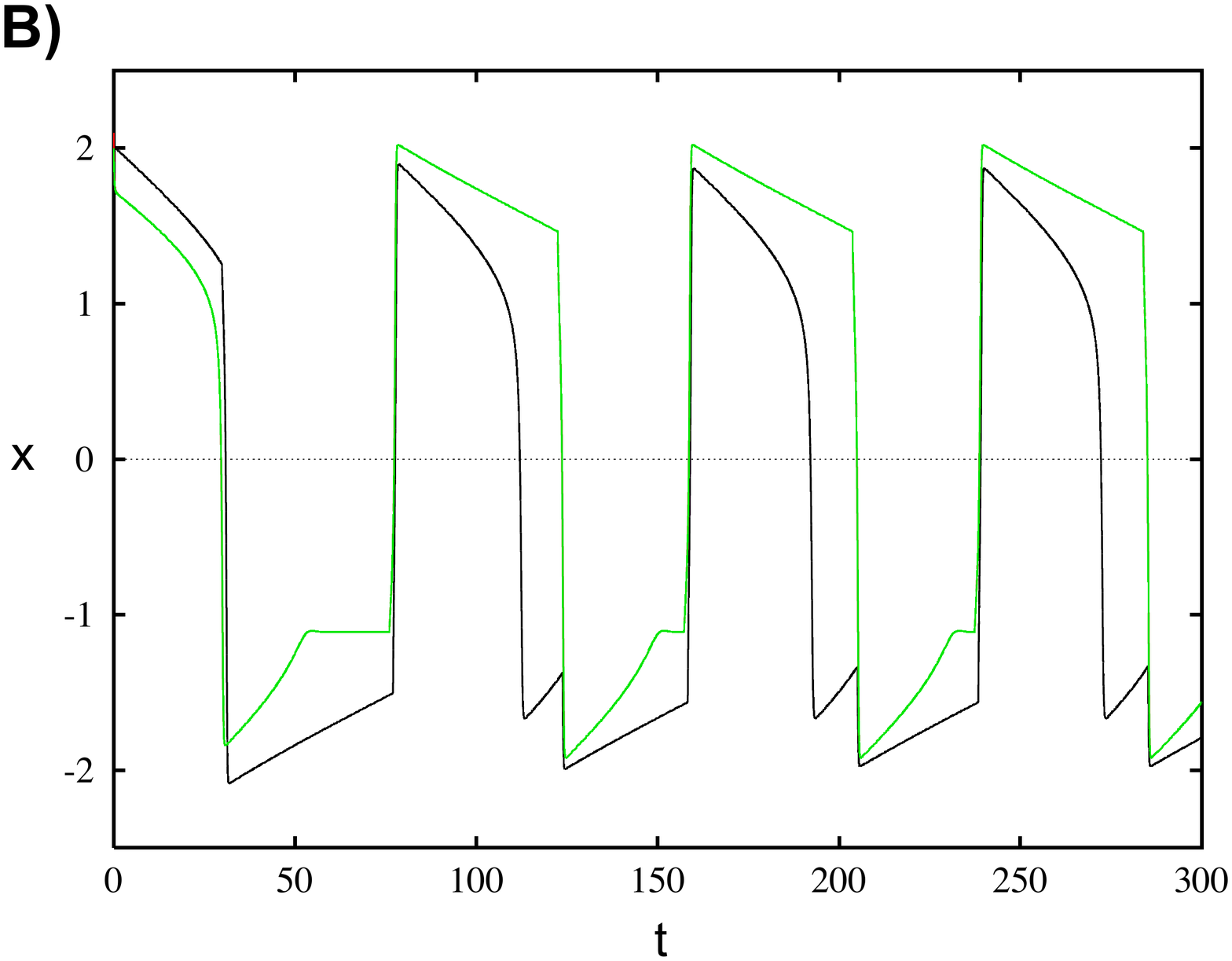}
\includegraphics[height=55mm, width=58mm]{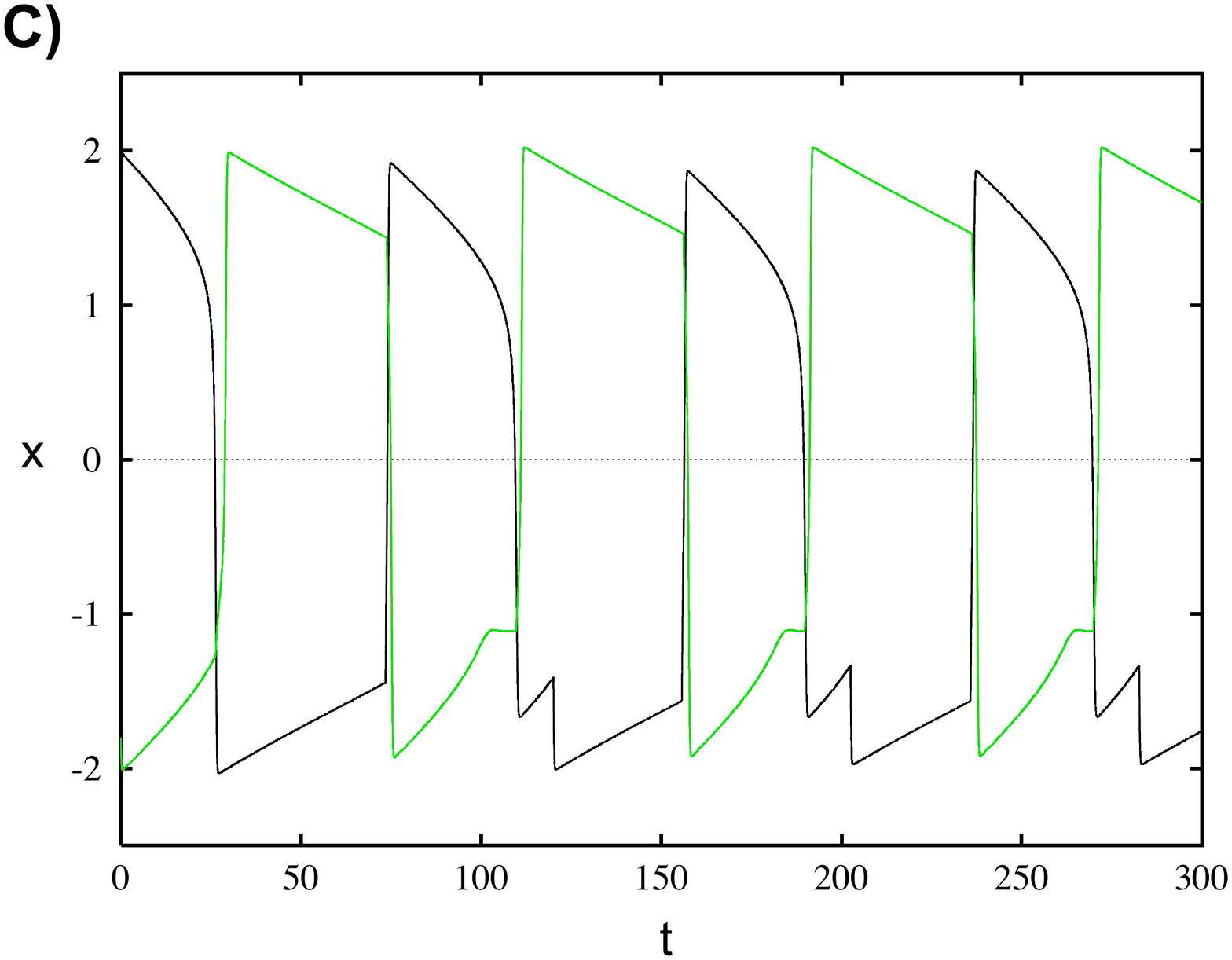}
\vspace{-0.5cm}
  \caption{Synchronous solutions for two $E$-cells and one $J$-cell with different combinations of the two coupling delays,  when the $E$-cells have a longer active phase; A) $\tau_J=30$ and $\tau_E=15$, B) $\tau_J=45$ and $\tau_E=0$, and C) $\tau_J=0$ and $\tau_E=45$. The black and green/red curves are the time courses of $J$-cell and $E$-cell voltages, respectively. Parameter values used are $\epsilon=0.025,~\gamma=\gamma_J=5,~\beta=\beta_J=10,~\delta=\delta_J=-1.1,~\lambda=2,~\lambda_J=-2,~\sigma=0.002$, and $\theta=-0.5$. Coupling parameter values are $g_{exc}=g_{inh}=0.5,~x_{inh}=-2.2$, and $x_{exc}=2.2$.  }
  \label{simple_longE}
\end{figure}

In contrast to Fig.~\ref{simple_longJ}, in all three solutions of Fig.~\ref{simple_longE}, the $E$-cells have longer active phase than the $J$-cell. First, in the oscillations shown in Fig.~\ref{simple_longE}A, the $J$-cell (black oscillations) fires $\tau_E$ time after the $E$-cells fire, and it jumps down right after the $E$-cells do. This corresponds to the case where $E$-cells immediately jump down when inhibition effectively turns on even before they reach their respective right knee, as described in Theorem~\ref{theorem_longE}.   Because of sufficiently large delay from the $J$-cell, it is possible that the $E$-cells do not receive inhibition by the time when the $J$-cell reaches close to the right knee. When inhibition finally turns on, the $E$-cells jump down, and then $J$-cell reaches its right knee and jumps down. After the $\tau_E$ time, the $J$-cell stays in the silent phase and jumps down to the left branch of another cubic (see Figs.~\ref{nullclines_longE}A and B) while the $E$-cells still receive inhibition and continue to move up in the silent phase. As soon as the $E$-cells are released from inhibition, they fire again, and one full cycle of the synchronous solution is complete. In Fig.~\ref{simple_longE}B, since there is no delay in the excitatory synapses (see Remark~\ref{remark_longE}), the $J$-cell fires immediately after the $E$-cells fire. The rest of trajectories are analogous to those of Fig.~\ref{simple_longE}A. Finally, Fig.~\ref{simple_longE}C shows the synchronous solution in case that there is no delay in the inhibitory synapses. In this case, the $E$-cells fire when the $J$-cell jumps down, resulting in the two populations being in opposite phases, as described in Corollary~\ref{corollary_longE_opposite}. 

\subsection{Thalamic model}\label{thalamic}
To verify whether our analysis of the effect of time delays in networks with global inhibition provided in Section~\ref{analysis} gives insight into more biologically relevant models, we consider a model for sleep spindle rhythms based on~\cite{GR94,Golomb94,TBK96,RT00b,RUBIN02}. The model is based on the Hodgkin Huxley formalism~\cite{HH} and builds on the work of~\cite{Golomb94,TBK96,RT00b,RUBIN02}. Model parameters and the forms of some nonlinear functions are given in~\ref{appendix}.

The model consists of two thalamocortical relay (TC) cells which are excitatory and one inhibitory cell which represents the global inhibition from the reticular nucleus, called a reticular cell (RE). The model based on~\cite{RT00b} has the following form
\begin{equation}
\begin{array}{rcl}
\displaystyle\frac{dV_{TC,i}}{dt} &=& -I_{ionic,TC}(V_{TC,i},h_{TC,i},r_{TC,i}) \\
&&- (V_{TC,i}-V_{inh})\left(g_{RT,A}\,s_{RT,A} + g_{RT,B}\,s_{RT,B}\right) \\
\displaystyle\frac{dh_{TC,i}}{dt}&=& \displaystyle\frac{h_{TC,\infty}(V_{TC,i})-h_{TC,i}}{\tau_{hTC}(V_{TC,i})}\\
\displaystyle\frac{dr_{TC,i}}{dt}&=& \displaystyle\frac{r_{TC,\infty}(V_{TC,i})-r_{TC,i}}{\tau_{rTC}(V_{TC,i})}\\
\displaystyle\frac{ds_{RT,A}}{dt}&=& \alpha_{RT,A}H_{s, \infty}(V_{RE,i}) (1-s_{RT,A})-\beta_{RT,A}s_{RT,A}
\end{array}
\label{thalamic_TC}
\end{equation}
Here $I_{ionic,TC}$ refers to the ionic currents present in the model for the TC cell. For the $i$th TC cell, this depends on the voltage, $V_{TC,i}$, and the gating variables for the currents $h_{TC,i}$ and $r_{TC,i}$ in the cell. Details of the ionic currents can be found in~\ref{appendix}.

Note that $g_{RT,A}$ corresponds the maximal conductance of the $GABA_{A}$ synapses (inhibitory A) from the RE cell to the TC cell and $s_{RT,A}$ is the corresponding gating variable. Similarly for the other synapse, $GABA_{B}$, corresponding to $g_{RT,B}$ and $s_{RT,B}$. The model for $s_{RT,B}$ is more complicated and  involves two differential equations, so we have not included it for simplicity. $g_{RT,A}$ and $V_{inh}$ represent the maximum conductance and the reversal potential, respectively, of the inhibitory synapse. More details concerning the biophysical significance of each term are described in~\cite{Golomb94,TBK96}.

The equations for each RE cell are:
\begin{equation}
\begin{array}{rcl}
\displaystyle\frac{dV_{RE,i}}{dt} &=& -I_{ionic,RE}(V_{RE,i},m_{RE,i},h_{RE,i}) \\
&&-g_{RR}(V_{RE,i}-V_{inh})\sum\limits_{i=1}^N\,s_{RR,i}
- g_{TR}(V_{RE,i}-V_{exc})\sum\limits_{i=1}^N \,s_{TR,i} \\
\displaystyle\frac{dm_{RE,i}}{dt}&=& \mu_1Ca_{RE}(1-m_{RE,i})-\mu_2m_{RE,i} \\
\displaystyle\frac{dh_{RE,i}}{dt}&=&\displaystyle\frac{h_{RE,\infty}(V_{RE,i})-h_{RE,i}}{\tau_{hRE}(V_{RE,i})}\\
\displaystyle\frac{dCa_{RE,i}}{dt} &=&-\nu I_{RT}(V_{RE,i},h_{RE,i})-\gamma Ca_{RE,i}\\
&&\\
\displaystyle\frac{ds_{RR,i}}{dt}&=& \alpha_{RR}H_{s, \infty}(V_{RE,i}) (1-s_{RR,i})-\beta_{RR}s_{RR,i}\\
\displaystyle\frac{ds_{TR,i}}{dt}&=&\alpha_{TR}H_{s, \infty}(V_{TC,i}) (1-s_{TR,i})-\beta_{TR}s_{TR,i}\\
\end{array}
\label{thalamic_RE}
\end{equation}

Similarly, $I_{ionic,RE}$ refers to the ionic currents present in the model for the $i$th RE cell, which depends on the voltage of the RE cell, $V_{RE,i}$, the gating variables for the currents, $m_{RE,i}$ and $h_{RE,i}$, and the calcium concentration in the RE cell, $Ca_{RE,i}$. The gating variables all following models of a similar form, based on the Hodgkin Huxley approach~\cite{HH}. The second term, $g_{RR}(V_{RE,i}-V_{inh})\sum_{i=1}^N\,s_{RR,i}$, represents the inhibitory input from other RE cells, referred to as self-inhibition. The last synaptic term, $g_{TR}(V_{RE,i}-V_{exc})\sum_{i=1}^N \,s_{TR,i}$, represents excitatory input from the TC cells where the sum is over all TC cells that send excitatory input to the $i$th RE cell. The synaptic variables $s_{RR,i}$ and $s_{TR,i}$ satisfy the last two equations given in (\ref{thalamic_RE}). $g_{TR}$ and $V_{exc}$ represent the maximum conductance and the reversal potential, respectively, of the excitatory synapse. 

To reduce the above model with arbitrary number of TC and RE cells to a model which is similar to what we have analyzed we get rid of all the differential equations for the $s$ variables, make the $s$ variables directly dependent on the voltages and add coupling delays to the voltages 
\begin{equation}
\begin{array}{rcl}
\displaystyle\frac{dV_{TC,i}}{dt} &=& -I_{ionic,TC}(V_{TC,i},r_{TC,i},h_{TC,i}) \\
&&- g_{RT,A}(V_{TC,i}-V_{inh})\,s_{RT}(V_{RE}(t-\tau_{RT}))\\
\displaystyle\frac{dh_{TC,i}}{dt}&=& \displaystyle\frac{h_{TC,\infty}(V_{TC,i})-h_{TC,i}}{\tau_{hTC}(V_{TC,i})}\\
&&\\
\displaystyle\frac{dr_{TC,i}}{dt}&=& \displaystyle\frac{r_{TC,\infty}(V_{TC,i})-r_{TC,i}}{\tau_{rTC}(V_{TC,i})}\\
\displaystyle\frac{dV_{RE}}{dt} &=& -I_{ionic,RE}(V_{RE},h_{RE},m_{RE}) \\
&& -g_{TR}(V_{RE}-V_{exc}) \sum\limits_{i=1}^N\,s_{TR,i}(V_{TC,i}(t-\tau_{TR})) \\
\displaystyle \frac{dm_{RE}}{dt}&=& \mu_1Ca_{RE}(1-m_{RE})-\mu_2m_{RE} \\
\displaystyle\frac{dh_{RE}}{dt}&=& \displaystyle\frac{h_{RE,\infty}(V_{RE})-m_{RE}}{\tau_{hRE}(V_{RE})}\\
\displaystyle\frac{dCa_{RE}}{dt} &=&-\nu I_{RT}(V_{RE},h_{RE})-\gamma Ca_{RE}\\
\end{array}
\label{thalamic_eqs}
\end{equation}
where the nonlinear functions $s_{RT}$ and $s_{TR}$ are sigmoidal using smooth approximations of Heaviside step functions with parameters based on the the synaptic equations in the model~\cite{RT00b}. Details are in~\ref{appendix}.

In the original model of Rubin and Terman~\cite{RT00b}, they argued that indirect 
synapses such as $GABA_B$ were necessary to stabilize the oscillations.
In our simple model of Section~\ref{simple}, this was not true if delays are present. Thus we have dropped the $GABA_B$ synapses in our reduced thalamic model. Further, for simplicity, we have dropped the self-inhibition from the RE cells, as this is not present in our simple model of the previous section. In fact, if we carried out some simulations with self-inhibition present, we found that oscillations are still possible. 
 
Figure~\ref{thalamicfig} shows the results of simulation of the model \eqref{thalamic_eqs} with parameters listed in~\ref{appendix} which are adapted from \cite{RT00b}. Note that the active phase of the RE cell is longer than that of the TC cell. The simulations of this model reproduce all the behaviors seen in the simple model (Section~\ref{simple}) when the active phase of the $J$ cell is longer than that of the $E$-cells.  In particular, in all cases the two TC cells synchronize. The relative size of the delays in the excitatory and inhibitory connections determines the phase relationship between the TC cells and the RE cell.

\begin{figure} [htb!]
\centering
\includegraphics[height=55mm, width=58mm]{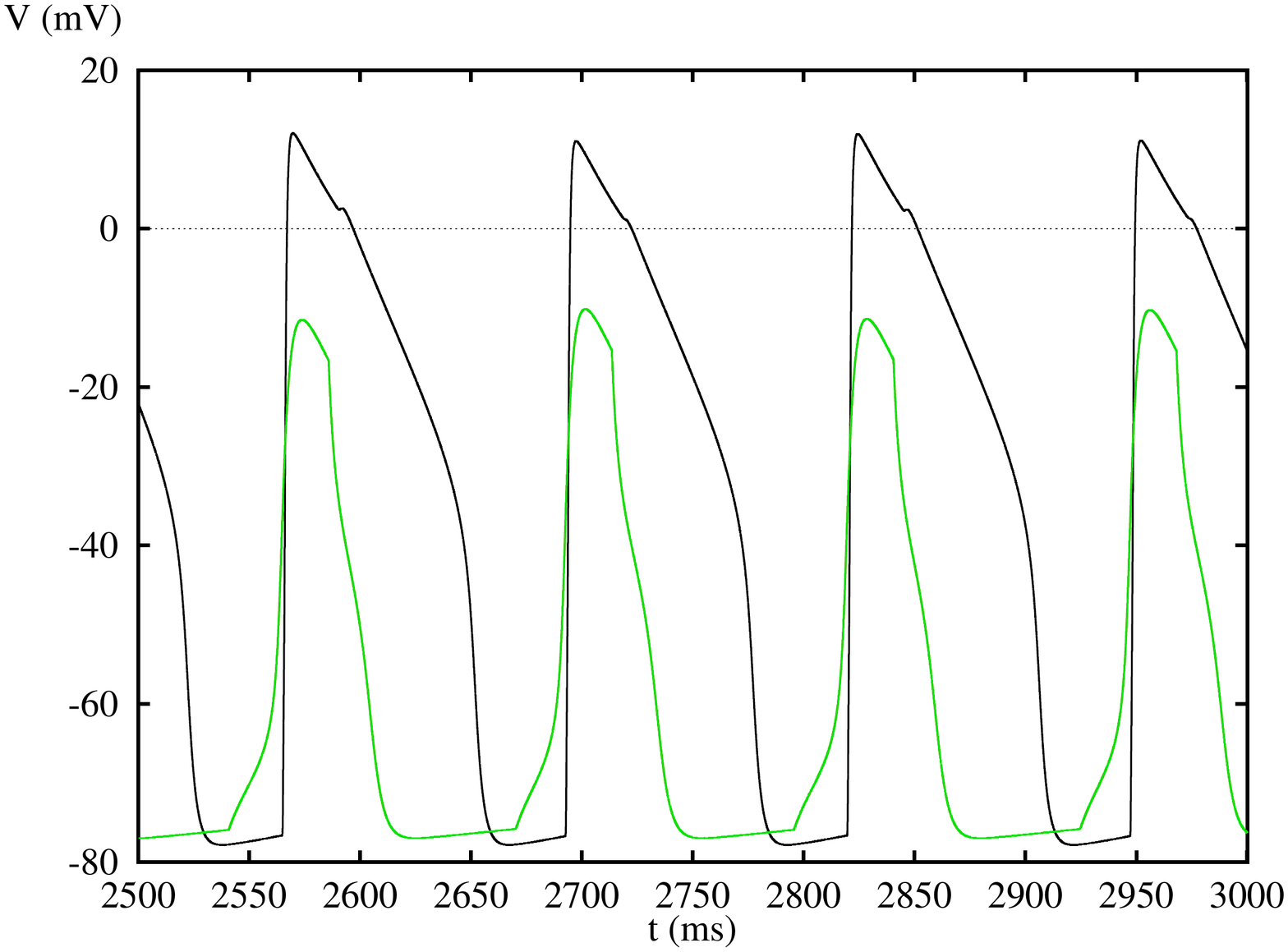}
\includegraphics[height=55mm, width=58mm]{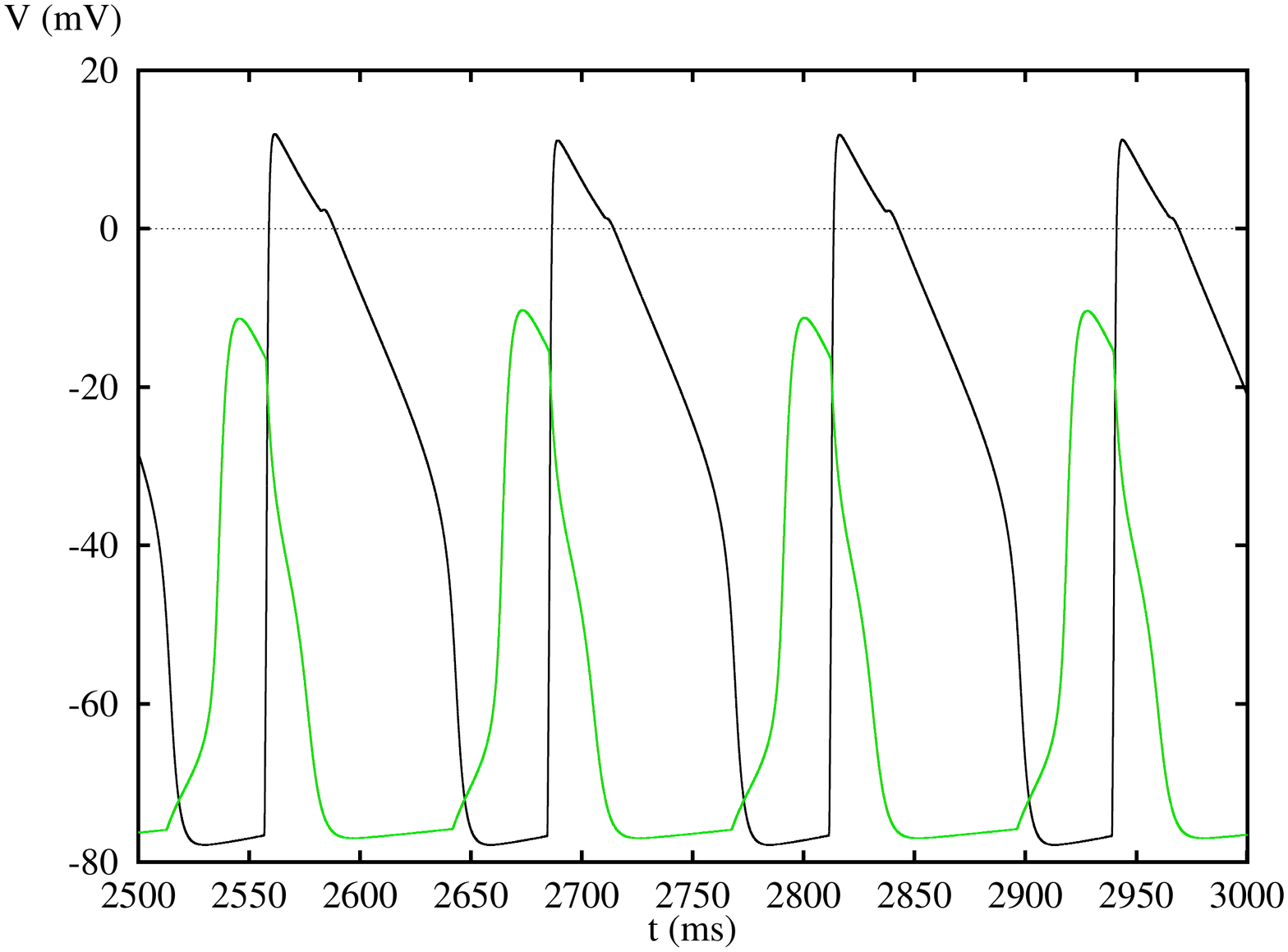}\\
\vspace{-0.5cm}
\includegraphics[height=55mm, width=58mm]{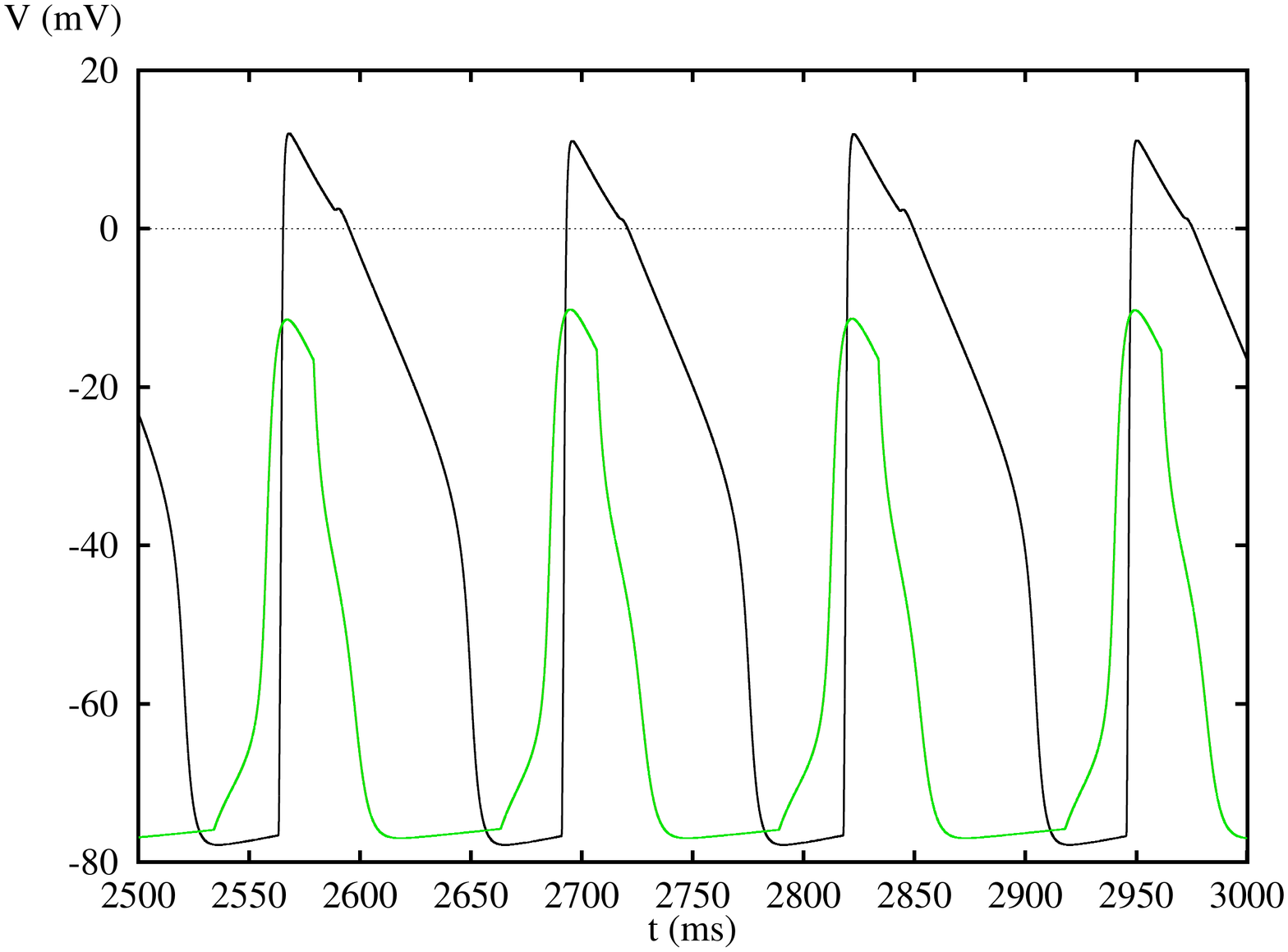}
\includegraphics[height=55mm, width=58mm]{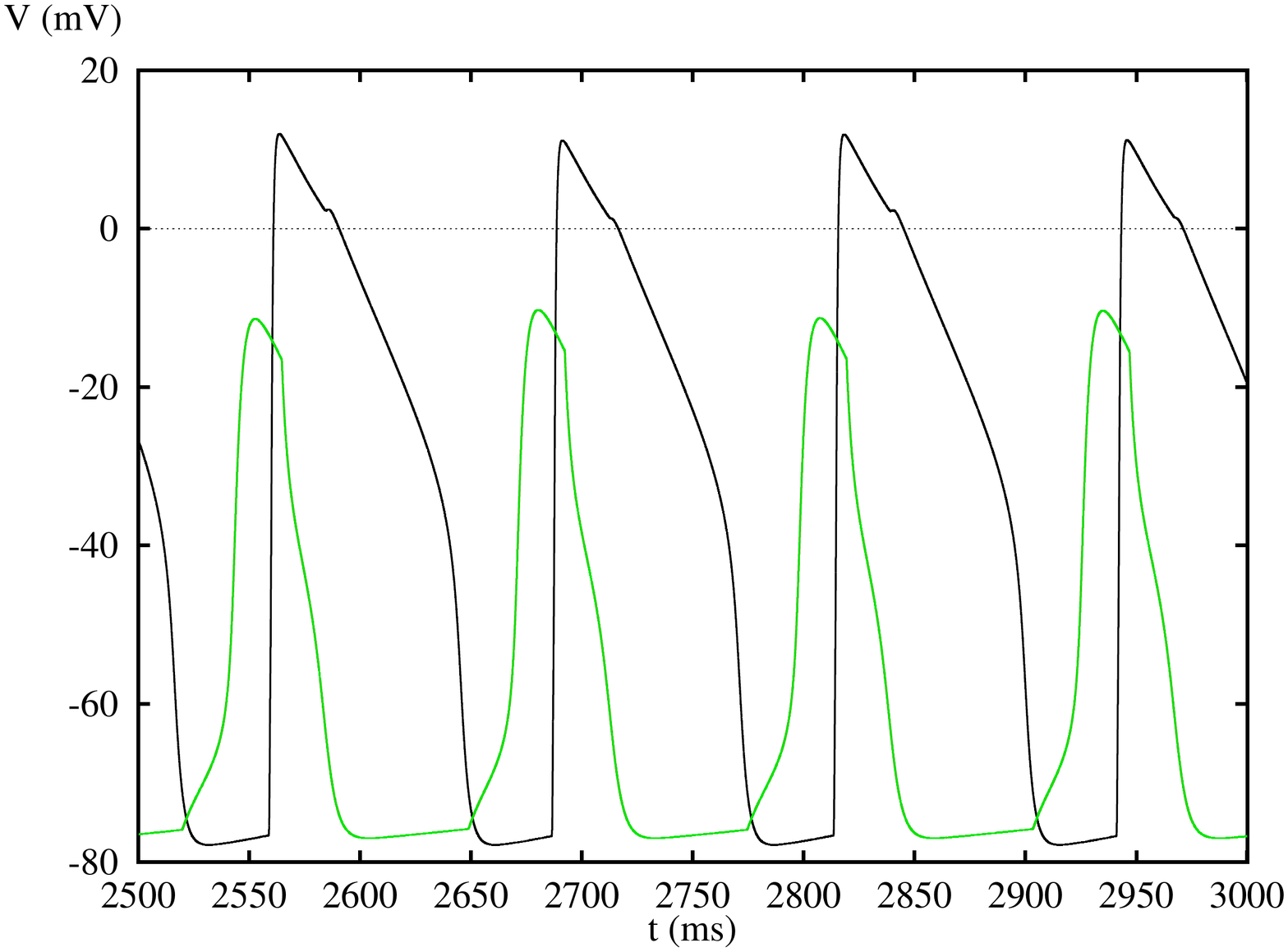}
\vspace{-0.5cm}
  \caption{Synchronous solutions for two TC cells and one RE cell for the thalamic model with different combinations of the two coupling delay. Single cell parameter values are used as in~\cite{RT00b} without self-inhibition on the RE cell and GABA$_B$ inhibition on TC cells.  A full list of parameter values is available in~\ref{appendix}. Other modification to model are described in the text. Coupling delays used are i) upper left figure: $\tau_{RT}=20$, $\tau_{TR}=0$; ii) upper right figure: $\tau_{RT}=0$, $\tau_{TR}=20$; iii) lower left figure: $\tau_{RT}=15$, $\tau_{TR}=5$; and iv) lower right figure: $\tau_{RT}=5$, $\tau_{TR}=15$. Black is the $RE$ cell, green and red are $TC$ cells.}
\label{thalamicfig}
\end{figure}

\section{Discussion} \label{discussion}
In this paper, we provide sufficient conditions for the existence and stability of synchronous solutions for globally inhibitory networks of oscillators with coupling delays. The model we develop and analyze is biologically motivated by several neural systems with this network structure. In the context of sleep rhythms, synchronization is one of the common rhythmic behaviors, in which excitatory thalamocortical relay cells fire together while receiving a global inhibition from a population of inhibitory thalamic reticular cells \cite{DS97}. In sensory processing, synchronization through global inhibition can 
be important for a network of excitatory neurons to produce the correct response to a given input \cite{doiron2003}.
Regardless of the relative duration of active phases for these two distinct populations, we showed that certain conditions imply the existence and stability of synchronous solutions, and the coupling delays play a significant role in the generation of synchronized network behaviors. 

In Section~\ref{analysis}, we apply geometric singular perturbation methods to prove existence and stability conditions in terms of the delays, $\tau_J$ corresponding to the delay in inhibitory synapses and $\tau_E$ corresponding to that in excitatory synapses. These delays represent both the time for information to travel 
between neurons and to be processed at the synapse.
We show that the presence of delays makes it possible for the network to exhibit synchronous solutions. In contrast, if there is no delay in both synapses, our analysis demonstrated that synchronous solutions cannot be obtained. Based on the construction of synchronous solutions under certain conditions, we consequently identify the period of such solutions in terms of the lengths of both delays. 

In related studies by Rubin and Terman~\cite{RT00,RUBIN02}, they assume that the inhibitory cells have a longer active phase than the excitatory cells, based on the experimental findings that thalamic reticular cells are known to have longer active states than relay cells~\cite{DS97}. However, as we indicate above, global inhibitory networks occur in other neural systems, where this assumption may or may not be true. Thus, we extend their analysis by considering the opposite case where excitatory cells have a longer active phase than the inhibitory cells, and show synchronous solutions can also exist as long as there are delays present. 

We provide numerical simulations using XPP~\cite{Ermentrout02} in Section~\ref{numerical} to supplement and
validate our analytical results provided in Section~\ref{analysis}. We specify explicit forms for the 
nonlinearities in our generic two-dimensional, relaxation oscillator model, Eqs.~(\ref{excx})--(\ref{inhy}),
which have appropriate form for the nullclines. The numerical simulations of this model confirm that the presence of the delay in either synapse is an essential factor to generate the synchronous solutions.  

One advantage of the explicit representation of coupling delays in model equations, unlike the models by Rubin and Terman~\cite{RT00,RUBIN02}, is that we can conduct a systematic study for the existence of solutions depending on the length of delays. This allows us to extend the way that the effect of synaptic delays were incorporated in previous models~\cite{RT00,RUBIN02} while having slightly simpler model equations. Thus, in Section~\ref{simple}, we consider three different combinations of inhibitory and excitatory delays for each case of which cell has the longer active phase. The model simulations demonstrate that all of the combinations result in synchronized behaviors among the excitatory cells. These synchronized oscillations differ from each other only in terms of the phase difference between excitatory and inhibitory cells' oscillations; their qualitative features are the same. 

In addition, our model extends the work on the effect of delays in ~\cite{CampWang98,FJWC01} in that ours 
analyzes a network of excitable neurons with both excitatory and inhibitory synapses whereas theirs 
focused on system with excitatory neurons which are oscillatory, that is, where each uncoupled neuron can oscillate without synaptic coupling. 

In Section~\ref{thalamic}, we consider a more biologically realistic model, based on a model for thalamic 
rhythms  \cite{GR94,Golomb94,TBK96,RT00b,RUBIN02}.  
This model is higher dimensional and thus does not fit the exact framework of our analysis, nevertheless
we observe that the presence of delay is, again, an essential factor to generate the synchronous solutions
of the excitatory cells.  This model even reproduces the relative phase relationships between the excitatory 
and inhibitory populations predicted by our analysis, depending on the relative sizes of delays that occur in 
the excitatory and inhibitory synapses. Analysis of this model would require looking at singularly perturbed
systems with higher dimensional slow subsystems, which is possible \cite{SW04,RW07}, but outside the scope 
of this article. These simulations are suggestive that the delay mechanisms we demonstrate in our analysis
of two dimensional relaxation oscillators carry over to systems with higher dimensional slow subsystems.

Our model extends some of previous modeling work on synaptic delays in biologically relevant neural networks. However,  analysis on other types of network behaviors, such as clustered patterns, is also needed to obtain a more complete understanding of how different population rhythms arise as a result of the interaction between coupling delays, intrinsic properties of each cell and network architecture. Also, since the present model assumes that $J$ population is nearly synchronized so that it can be viewed as a single cell, we can relax this condition by allowing the interaction between neurons in the $J$ population, which may result in different population behaviors other than synchronization, such as clustering. Investigating the role of the interactions between inhibitory thalamic reticular cells in the thalamocortical networks is worthy of further investigation in the context of network firing patterns.
\vspace{-0.2cm}
\appendix \section{Model Equations and Parameters}\label{appendix}
The equations for thalamocortical relay (TC) and reticular (RE) cells in the thalamic network are given in Section~\ref{thalamic}. We first provide model forms of nonlinear functions for the original thalamic model described in (\ref{thalamic_TC}) and (\ref{thalamic_RE}), and then model parameters used in (\ref{thalamic_eqs}) to generate Figure~\ref{thalamicfig}. As our biophysical neuronal network of excitable cells, we considered a pair of TC cells, each governed by the first three equations in (\ref{thalamic_eqs}), and one RE cell which satisfies the rest four equations of (\ref{thalamic_eqs}). For simplicity the index $i$ for TC cells is omitted in the equations that describe them.

\vspace{-0.1cm}
\subsection*{\it \small TC Cells}
The ionic current term, $I_{ionic,TC}$, in (\ref{thalamic_TC}) consist of three components:
\begin{align*}
I_{ionic,TC}=I_{TT}(V_{TC,i}, h_{TC,i})+I_{sag}(V_{TC,i}, r_{TC,i})+I_L(V_{TC,i})
\end{align*}

\noindent $I_{TT}$:
\begin{align*}
I_{TT}(V, h)&=g_{Ca,TC}m_{TC, \infty}^2(V)h(V-V_{Ca,TC}) \\
m_{TC, \infty}(V)&=[1+\exp(-(V-\theta_{tm})/\sigma_{tm})]^{-1}\\
h_{TC, \infty}(V)&=[1+\exp(-(V-\theta_{th})/\sigma_{th})]^{-1}\\
\tau_{hTC}(V)&=\tau_{h0}+\tau_{h1}[1+\exp(-(V-\theta_{\tau h})/\sigma_{\tau h})]^{-1}
\end{align*}
where all model parameters appeared in the equations above can be adapted from \cite{Golomb94,RT00b} for different types of solutions such as synchronization and clustering.

\noindent $I_{sag}$ and $I_{L}$:
\begin{align*}
I_{sag}(V, r)&=g_{sag}r(V-V_{sag}) \\
r_{TC, \infty}(V)&=[1+\exp(-(V-\theta_r)/\sigma_r)]^{-1}\\
\tau_{rTC}(V)&=\tau_{r0}+\tau_{r1}[\exp(-(V-\theta_{\tau r0})/\sigma_{\tau r0})+\exp(-(V-\theta_{\tau r1})/\sigma_{\tau r1})]^{-1}\\
I_{L}(V)&=g_{L}(V-V_{L})
\end{align*}

In the synaptic term, $s_{RT,A}$, $H_{s, \infty}$ represents the approximation of the Heaviside function with the form
\begin{align*}
H_{s, \infty}(V)&=[1+\exp(-(V-\theta_{s})/\sigma_{s})]^{-1}
\end{align*}
\vspace{-0.1cm}

\subsection*{\it \small RE Cell}
The ionic current term, $I_{ionic,RE}$, in (\ref{thalamic_RE}) consist of three components:
\begin{align*}
I_{ionic,RE}=I_{RT}(V_{RE}, h_{RE})+I_{AHP}(V_{RE}, m_{RE})+I_{RL}(V_{RE})
\end{align*}

\noindent $I_{RT}$:
\begin{align*}
I_{RT}(V, h)&=g_{Ca,RE}m_{RE, \infty}^2(V)h(V-V_{Ca,RE}) \\
m_{RE, \infty}(V)&=[1+\exp(-(V-\theta_{rm})/\sigma_{rm})]^{-1}\\
h_{RE, \infty}(V)&=[1+\exp(-(V-\theta_{rh})/\sigma_{rh})]^{-1}\\
\tau_{rRE}(V)&=\tau_{r0}+\tau_{r1}[1+\exp(-(V-\theta_{\tau rh})/\sigma_{\tau rh})]^{-1}
\end{align*}
The specific model parameters used to show different network behaviors are given in \cite{Golomb94,RT00b}. Similar to $s_{RT,A}$ term, the equations for the synaptic terms, $s_{RR,i}$ and $s_{TR,i}$, in each $i$th TC cell use the approximations of Heaviside step function.

\noindent $I_{AHP}$ and $I_{RL}$:
\begin{align*}
I_{AHP}(V, m_{RE})&=g_{AHP}m_{RE}(V-V_{K}) \\
I_{RL}(V)&=g_{RL}(V-V_{RL})
\end{align*}

For the reduced model described in (\ref{thalamic_eqs}), the specific form for the synaptic variable, $s_{i}, i=RT, TR$, is given by
\begin{align*}
s_{i}(V)&=[1+\exp(-(V-\theta_{i})/\sigma_i)]^{-1}.
\end{align*}

To generate synchronized solutions among two TC cells receiving global inhibition from one RE cell in the reduced model (\ref{thalamic_eqs}) displayed in Figure~\ref{thalamicfig}, we used the following parameter values based on Refs.~\cite{Golomb94} and \cite{RT00b}. Two TC cells: $I_{TT}$: $g_{Ca,TC}=1.5$~mS/cm$^2$, $\theta_{tm}=-59.0$~mV, $\sigma_{tm}=9.0$~mV, $V_{Ca,TC}=90$~mV, $\theta_{th}=-82$~mV, $\sigma_{th}=-5.0$~mV, $\tau_{h0}=66.\bar{6}$~ms, $\tau_{h1}=333.\bar{3}$~ms, $\theta_{\tau h}=-78.0$~mV, $\sigma_{\tau h}=-1.5$~mV; $I_{sag}$: $g_{sag}=0.15$~mS/cm$^2$,  $V_{sag}=-40$~mV, $\theta_r=-75.0$~mV, $\sigma_r=-5.5$~mV, $\tau_{h0}=20.0$~ms, $\tau_{h1}=1000.0$~ms, $\theta_{\tau h0}=-71.5$~mV, $\sigma_{\tau h0}=-14.2$~mV, $\theta_{\tau r1}=-89$~mV, $\sigma_{\tau r1}=-11.6$~mV; $I_{L}$: $g_L=0.2$~mS/cm$^2$, $V_L=-76.0$~mV; $s_{RT}$: 
$g_{RT, A}=0.1$~mS/cm$^2$, $V_{inh}=-84$~mV, $\theta_{RT}=-50$~mV, $\sigma_{RT}=0.5$~mV. 

One RE cell: $I_{RT}$: $g_{Ca,RE}=2.0$~mS/cm$^2$, $\theta_{rm}=-52.0$~mV, $\sigma_{rm}=9.0$~mV, $V_{Ca,RE}=90$~mV, $\theta_{rh}=-72$~mV, $\sigma_{rh}=-2.0$~mV, $\tau_{r0}=66.\bar{6}$~ms,  $\tau_{r1}=333.\bar{3}$~ms, $\theta_{\tau rh}=-78.0$~mV, $\sigma_{\tau rh}=-1.0$~mV; $I_{AHP}$: $g_{AHP}=0.1$~mS/cm$^2$,  $V_{K}=-90$~mV, $\mu_1=0.02$~1/ms, $\mu_2=0.025$~1/ms, $\nu=0.01$~1/ms, $\gamma=0.08$~1/ms; $I_{RL}$: $g_{RL}=0.3$~mS/cm$^{2}$, $V_{RL}=-76.0$~mV; $s_{TR}$: $g_{TR}=0.6$~mS/cm$^2$, $V_{exc}=0$~mV, $\theta_{TR}=-35$~mV, $\sigma_{TR}=0.5$~mV.



\vspace{0.5cm}
\noindent {\bf Acknowledgements.}
{\small The first author was supported by a University of Hartford Greenberg Junior Faculty Grant. This support does not necessarily imply endorsement by the University of Hartford of project conclusions. The second author was supported by a grant from the Natural Sciences and Engineering Research Council of Canada. }

\clearpage
  \bibliographystyle{spmpsci}  
  \bibliography{refs2}

\begin{thebibliography}{10}
\providecommand{\url}[1]{{#1}}
\providecommand{\urlprefix}{URL }
\expandafter\ifx\csname urlstyle\endcsname\relax
  \providecommand{\doi}[1]{DOI~\discretionary{}{}{}#1}\else
  \providecommand{\doi}{DOI~\discretionary{}{}{}\begingroup
  \urlstyle{rm}\Url}\fi

\bibitem{Barton2006}
Barton, D.A., Krauskopf, B., Wilson, R.E.: Periodic solutions and their
  bifurcations in a non-smooth second-order delay differential equation.
\newblock Dynamical Systems \textbf{21}(3), 289--311 (2006)

\bibitem{bezaire2013}
Bezaire, M.J., Soltesz, I.: Quantitative assessment of {CA}1 local circuits:
  knowledge base for interneuron-pyramidal cell connectivity.
\newblock Hippocampus \textbf{23}(9), 751--785 (2013)

\bibitem{BT}
Buri{\'c}, N., Todorovi{\'c}, D.: Dynamics of {F}itzhugh-{N}agumo excitable
  systems with delayed coupling.
\newblock Phys. Rev. E \textbf{67}, 066--222 (2003)

\bibitem{buzsaki94}
Buzs{\'a}ki, G., Llin{\'a}s, R., Singer, W., Berthoz, A., Chrtisten, Y. (eds.):
  Temporal coding in the brain.
\newblock Springer--Verlag, New York, NY (1994)

\bibitem{CampWang17}
Campbell, S.A., Wang, Z.: Phase models and clustering in networks of
  oscillators with delayed coupling.
\newblock Physica D: Nonlinear Phenomena  (2017).
\newblock In press

\bibitem{CampWang98}
Campbell, S.R., Wang, D.: Relaxation oscillators with time delay coupling.
\newblock Physica D: Nonlinear Phenomena \textbf{111}(1), 151--178 (1998)

\bibitem{choe2010}
Choe, C.U., Dahms, T., H{\"o}vel, P., Sch{\"o}ll, E.: Controlling synchrony by
  delay coupling in networks: From in-phase to splay and cluster states.
\newblock Phys. Rev. E \textbf{81}(2), 025,205 (2010)

\bibitem{CLMP89}
Chow, S.N., Lin, X.B., Mallet-Paret, J.: Transition layers for singularly
  perturbed delay differential equations with monotone nonlinearities.
\newblock J. Dynam. Differential Equations \textbf{1}(1), 3--43 (1989)

\bibitem{CDSS97}
Contreras, D., Destexhe, A., Sejnowski, T.J., Steriade, M.: Spatiotemporal
  patterns of spindle oscillations in cortex and thalamus.
\newblock Journal of Neuroscience \textbf{17}(3), 1179--1196 (1997)

\bibitem{Crook97}
Crook, S., Ermentrout, G., Vanier, M., Bower, J.: The role of axonal delay in
  synchronization of networks of coupled cortical oscillators.
\newblock JCN \textbf{4}, 161--172 (1997)

\bibitem{DLS12}
Dahms, T., Lehnert, J., Sch{\"o}ll, E.: Cluster and group synchronization in
  delay-coupled networks.
\newblock Phys. Rev. E \textbf{86}(1), 016,202 (2012)

\bibitem{DBMS96}
Destexhe, A., Bal, T., McCormick, D.A., Sejnowski, T.J.: Ionic mechanisms
  underlying synchronized oscillations and propagating waves in a model of
  ferret thalamic slices.
\newblock J. Neurophysiol. \textbf{76}, 2049--2070 (1996)

\bibitem{Destexhe98}
Destexhe, A., Mainen, Z., Sejnowski, T.: Kinetic models of synaptic
  transmission.
\newblock In: C.~Koch, I.~Segev (eds.) Methods in Neuronal Modeling: From
  Synapses to Networks, chap.~1. MIT Press, Cambridge, MA (1998)

\bibitem{DMS93}
Destexhe, A., McCormick, D.A., Sejnowski, T.J.: A model for 8--10 {H}z
  spindling in interconnected thalamic relay and reticularis neurons.
\newblock Biophys. J. \textbf{65}, 2474--2478 (1993)

\bibitem{DS97}
Destexhe, A., Sejnowski, T.J.: Synchronized oscillations in thalamic networks:
  Insights from modeling studies.
\newblock In: M.~Steriade, E.G. Jones, D.A. McCormick (eds.) Thalamus.
  Elsevier, Amsterdam (1997)

\bibitem{doiron2003}
Doiron, B., Chacron, M.J., Maler, L., Longtin, A., Bastian, J.: Inhibitory
  feedback required for network oscillatory responses to communication but not
  prey stimuli.
\newblock Nature \textbf{421}(6922), 539--543 (2003)

\bibitem{douglas2007}
Douglas, R.J., Martin, K.A.: Recurrent neuronal circuits in the neocortex.
\newblock Current biology \textbf{17}(13), R496--R500 (2007)

\bibitem{Ermentrout02}
Ermentrout, B.: Simulating, analyzing, and animating dynamical systems: a guide
  to {XPPAUT} for researchers and students, vol.~14.
\newblock SIAM (2002)

\bibitem{FJWC01}
Fox, J.J., Jayaprakash, C., Wang, D., Campbell, S.R.: Synchronization in
  relaxation oscillator networks with conduction delays.
\newblock Neural Computation \textbf{13}(5), 1003--1021 (2001)

\bibitem{Fridman2002}
Fridman, E.: Effects of small delays on stability of singularly perturbed
  systems.
\newblock Automatica \textbf{38}(5), 897--902 (2002)

\bibitem{GR94}
Golomb, D., Rinzel, J.: Clustering in globally coupled inhibitory neurons.
\newblock Physica D \textbf{72}, 259--282 (1994)

\bibitem{Golomb94}
Golomb, D., Wang, X.J., Rinzel, J.: Synchronization properties of spindle
  oscillations in a thalamic reticular nucleus model.
\newblock J. Neurophysiol \textbf{72}(3), 1109--1126 (1994)

\bibitem{Golomb96}
Golomb, D., Wang, X.J., Rinzel, J.: Propagation of spindle waves in a thalamic
  slice model.
\newblock J. Neurophysiol \textbf{75}(2), 750--769 (1996)

\bibitem{HH}
Hodgkin, A., Huxley, A.: A quantitative description of membrane current and its
  application to conduction and excitation in nerve.
\newblock J. Physiology \textbf{117}, 500--544 (1952)

\bibitem{Jacklet89}
Jacklet, J. (ed.): Neuronal and cellular oscillators.
\newblock Marcel Dekker Inc, New York, NY (1989)

\bibitem{KPR}
Kim, S., Park, S.H., Ryu, C.: Multistability in coupled oscillator systems with
  time delay.
\newblock Phys. Rev. Lett. \textbf{79}, 2911--2914 (1997)

\bibitem{KimBal95}
Kim, U., Bal, T., McCormick, D.: Spindle waves are propagating synchronized
  oscillations in the ferret {LGN}d in vitro.
\newblock J. Neurophysiol \textbf{74}(3), 1301--1323 (1995)

\bibitem{KL94}
Kopell, N., Le{M}asson, G.: Rhythmogenesis, amplitude modulation, and
  multiplexing in a cortical architecture.
\newblock Proc Natl Acad Sci USA \textbf{91}(22), 10,586--10,590 (1994)

\bibitem{Linas88}
Llin{\'a}s, R.R.: The intrinsic electrophysiological properties of mammalian
  neurons: insights into central nervous system function.
\newblock Science \textbf{242}(4886), 1654--1664 (1988)

\bibitem{LoFaro99}
LoFaro, T., N., K.: Timing regulation in a network reduced from voltage-gated
  equations to a one-dimensional map.
\newblock J Math Biol \textbf{38}(6), 479--533 (1999)

\bibitem{Luz}
Luzyanina, T.: Synchronization in an oscillator neural network model with
  time-delayed coupling.
\newblock Network: Computation in Neural Systems \textbf{6}, 43--59 (1995)

\bibitem{MPN86}
Mallet-Paret, J., Nussbaum, R.D.: Global continuation and asymptotic behaviour
  for periodic solutions of a differential-delay equation.
\newblock Annali di Matematica Pura ed Applicata \textbf{145}(1), 33--128
  (1986)

\bibitem{MRTWBC}
Miller, J., Ryu, H., Teymuroglu, Z., Wang, X., Booth, V., Campbell, S.:
  Clustering in inhibitory neural networks with nearest neighbor coupling.
\newblock In: T.~Jackson, A.~Radunskaya (eds.) Applications of Dynamical
  Systems in Biology and Medicine, pp. 99--121. Springer, New York (2015)

\bibitem{Mish80}
Mishchenko, E.F., Rozov, N.K.: Differential equations with small parameters and
  relaxation oscillations, vol.~13.
\newblock Springer, New York, NY (1980)

\bibitem{OroszSIADS14}
Orosz, G.: Decomposition of nonlinear delayed networks around cluster states
  with applications to neurodynamics.
\newblock SIAM J. Appl. Dyn. Syst. \textbf{13}(4), 1353--1386 (2014)

\bibitem{poo2009}
Poo, C., Isaacson, J.S.: Odor representations in olfactory cortex: ``sparse"€
  coding, global inhibition, and oscillations.
\newblock Neuron \textbf{62}(6), 850--861 (2009)

\bibitem{Rinzel98}
Rinzel, J., Terman, D., Wang, X.J., Ermentrout, B.: Propagating activity
  patterns in large-scale inhibitory neuronal networks.
\newblock Science \textbf{279}(5355), 1351--1355 (1998)

\bibitem{roux2015}
Roux, L., Buzs{\'a}ki, G.: Tasks for inhibitory interneurons in intact brain
  circuits.
\newblock Neuropharmacology \textbf{88}, 10--23 (2015)

\bibitem{RT00}
Rubin, J.E., Terman, D.: Analysis of clustered firing patterns in synaptically
  coupled networks of oscillators.
\newblock J. Math. Biol. \textbf{41}, 513--545 (2000)

\bibitem{RT00b}
Rubin, J.E., Terman, D.: Geometric analysis of population rhythms in
  synaptically coupled neuronal networks.
\newblock Neural Computation \textbf{12}(3), 597--645 (2000)

\bibitem{RUBIN02}
Rubin, J.E., Terman, D.: Geometric singular perturbation analysis of neuronal
  dynamics.
\newblock Handbook of Dynamical Systems \textbf{2}, 93--146 (2002)

\bibitem{RW07}
Rubin, J.E., Wechselberger, M.: Giant squid-hidden canard: the 3{D} geometry of
  the {H}odgkin-{H}uxley model.
\newblock Biological cybernetics \textbf{97}(1), 5--32 (2007)

\bibitem{SSA11}
Sethia, G.C., Sen, A., Atay, F.M.: Phase-locked solutions and their stability
  in the presence of propagation delays.
\newblock Pramana \textbf{77}(5), 905--915 (2011)

\bibitem{Sieber2006}
Sieber, J.: Dynamics of delayed relay systems.
\newblock Nonlinearity \textbf{19}(11), 2489 (2006)

\bibitem{Sieber2010}
Sieber, J., Kowalczyk, P., Hogan, S., Di~Bernardo, M.: Dynamics of symmetric
  dynamical systems with delayed switching.
\newblock Journal of Vibration and Control \textbf{16}(7-8), 1111--1140 (2010)

\bibitem{skinner94}
Skinner, F., Kopell, N., Marder, E.: Mechanisms for oscillation and frequency
  control in reciprocally inhibitory model neural networks.
\newblock J Comput Neurosci \textbf{1}, 69--87 (1994)

\bibitem{Somers93}
Somers, D., Kopell, N.: Rapid synchronization through fast threshold
  modulation.
\newblock Biol Cybern \textbf{68}(5), 393--407 (1993)

\bibitem{steriade90}
Steriade, M., Jones, E.G., Llin{\'a}s, R.R.: Thalamic oscillations and
  signaling.
\newblock Wiley, New York, NY (1990)

\bibitem{SMS93}
Steriade, M., McCormick, D.A., Sejnowski, T.J.: Thalamocortical oscillations in
  the sleeping and aroused brain.
\newblock Science \textbf{262}, 679--685 (1993)

\bibitem{Sun2006}
Sun, X.M., Zhao, J., Hill, D.J.: Stability and ${L}_2$-gain analysis for
  switched delay systems: A delay-dependent method.
\newblock Automatica \textbf{42}(10), 1769--1774 (2006)

\bibitem{SW04}
Szmolyan, P., Wechselberger, M.: Relaxation oscillations in {R}{$^3$}.
\newblock Journal of Differential Equations \textbf{200}(1), 69--104 (2004)

\bibitem{TBK96}
Terman, D., Bose, A., Kopell, N.: Functional reorganization in thalamocortical
  networks: Transition between spindling and delta sleep rhythms.
\newblock Proc. Natl. Acad. Sci. USA \textbf{93}, 15,417--15,422 (1996)

\bibitem{TEY01}
Terman, D., Ermentrout, G., Yew, A.: Propagating activity patterns in thalamic
  neuronal networks.
\newblock SIAM J. Appl. Math. \textbf{61}(5), 1578--1604 (2001)

\bibitem{TK98}
Terman, D., Kopell, N., Bose, A.: Dynamics of two mutually coupled inhibitory
  neurons.
\newblock Physica D \textbf{117}, 241--275 (1998)

\bibitem{TL97}
Terman, D., Lee, E.: Partial synchronization in a network of neural
  oscillators.
\newblock SIAM J. Appl. Math. \textbf{57}, 252--293 (1997)

\bibitem{terman95}
Terman, D., Wang, D.: Global competition and local cooperation in a network of
  neural oscillators.
\newblock Physica D: Nonlinear Phenomena \textbf{81}(1), 148--176 (1995)

\bibitem{tomasi2012}
Tomasi, S., Caminiti, R., Innocenti, G.M.: Areal differences in diameter and
  length of corticofugal projections.
\newblock Cerebral Cortex \textbf{22}(6), 1463--1472 (2012)

\bibitem{traub91}
Traub, R.D., Miles, R.: Neuronal networks of the hippocampus.
\newblock Cambridge University Press, New York, NY (1991)

\bibitem{wang92}
Wang, X.J., Rinzel, J.: Alternating and synchronous rhythms in reciprocally
  inhibitory model neurons.
\newblock Neural computation \textbf{4}, 84--97 (1992)

\end{thebibliography}





\end{document}